\documentclass[12pt]{amsart}
\usepackage[margin=1in]{geometry}
\usepackage{amsfonts}
\usepackage{amssymb}
\usepackage[dvips]{graphics}
\usepackage{epsfig}
\usepackage{euscript}
\usepackage{color}

 \newtheorem{thm}{Theorem}[section]
 
 \newtheorem{lemma}[thm]{Lemma}
 \newtheorem{prop}[thm]{Proposition}
 \theoremstyle{definition}
 
 \theoremstyle{remark}

 \numberwithin{equation}{section}

\newcommand{\be}{\begin{equation}}
\newcommand{\ee}{\end{equation}}
\newcommand{\bea}{\begin{eqnarray}}
\newcommand{\eea}{\end{eqnarray}}

\newcommand{\abs}[1]{\left\vert #1 \right\vert}
\newcommand{\conj}[1]{\overline{#1}}
\newcommand{\norm}[1]{\left\Vert#1\right\Vert}
\newcommand{\set}[1]{\left\{ #1\right\} }
\newcommand{\ip}[1]{\left\langle #1 \right\rangle}
\newcommand{\C}{\mathbb{C}}
\newcommand{\R}{\mathbb{R}}

\newcommand{\Z}{\mathbb{Z}}

\newcommand{\E}{{\mathbb E}}

\newcommand{\N}{\mathbb{N}}
\newcommand{\UU}{\mathbb{U}}
\newcommand{\PP}{\mathbb{P}}

\newcommand{\eps}{\varepsilon}

\DeclareMathOperator{\tr}{tr}

\DeclareMathOperator{\supp}{supp}

\DeclareMathOperator{\Cl}{Cl}

\DeclareMathOperator{\SL}{SL}
\DeclareMathOperator{\Sp}{Sp}

\definecolor{Plum}{rgb}{.5,0,1}

\numberwithin{equation}{section}

\begin{document}

\title{Localization for random block operators\\ related to the XY spin chain}

\author[J. Chapman]{Jacob Chapman$^1$}
\thanks{Both authors were supported in part by NSF grant DMS-1069320 (PI G.\ Stolz).}
\address{$^1$ Department of Mathematics\\ 
William Carey University\\
Hattiesburg, MS 39401, USA}
\email{jchap12286@gmail.com}

\author[G. Stolz]{G\"unter Stolz$^2$}
\address{$^2$ Department of Mathematics\\
University of Alabama at Birmingham\\
Birmingham, AL 35294, USA}
\email{stolz@math.uab.edu}

\date{}

%\vspace{.3truein}
%\centerline{\bf Abstract}\mdeskip
%
\begin{abstract}

We study a class of random block operators which appear as effective one-particle Hamiltonians for the anisotropic XY quantum spin chain in an exterior magnetic field given by an array of i.i.d.\ random variables. For arbitrary non-trivial single-site distribution of the magnetic field, we prove dynamical localization of these operators at non-zero energy.

\end{abstract}

\maketitle

%%%%%%%%%%%%%%%%%%%%%%%%%%%%%%%%%%%%%%%%%%%%%%%%%%%%%%%%%%

\section{Introduction} \label{sec:intro}
\markboth{\scriptsize J.\ CHAPMAN AND G.\ STOLZ}{\scriptsize LOCALIZATION FOR RANDOM BLOCK OPERATORS RELATED TO THE XY SPIN CHAIN}
It was first understood in the groundbreaking paper \cite{LSM1961} of Lieb, Schultz, and Mattis in 1961 that the XY spin chain provides one of very few exactly solvable models in the theory of interacting quantum many-body systems. Technically, this is due to the fact that the Jordan-Wigner transform maps the XY chain Hamiltonian to a free Fermion system, i.e.\ a Hamiltonian quadratic in a set of creation and annihilation operators satisfying canonical anti-commutation relations. This works for the isotropic as well as the anisotropic XY chain and also allows to include a transversal exterior magnetic field. It is not necessary that the strength of the spin-couplings and the field are constant along the chain. However, in the variable coefficient case the Jordan-Wigner transform does not lead to an exact solution (diagonalization) of the XY chain, but reduces it to the study of an effective one-particle Hamiltonian.

Consequentially, for the last half century the XY chain has become a prototypical model, often used as a starting point in attempts to understand phenomena in many-body quantum theory, before moving on to the investigation of more complicated (and more realistic) models. This is true, in particular, for attempts to understand the effect of {\it disorder} on many-body quantum systems. In this context a key issue is to learn how to characterize and prove {\it many-body localization}, a concept which is still under close scrutiny in the physics literature, e.g.\ \cite{Baskoetal, OH, ZPP, PalHuse}. 

The work we present here was inspired by the paper \cite{HSS2012} which gave a rigorous proof of dynamical localization for the XY chain in random exterior field, using the concept of a zero-velocity Lieb-Robinson bound to characterize dynamical localization in quantum spin systems.  The main result in \cite{HSS2012} is for the {\it isotropic} XY chain in random field, as in this case the effective one-particle Hamiltonian found via Jordan-Wigner is the Anderson model. Many-body dynamical localization for the disordered XY chain thus arises as a consequence of the strong one-particle dynamical localization bounds known for the Anderson model.

Here we will be concerned with the {\it anisotropic} XY chain in random field. As essentially already contained in \cite{LSM1961} (and briefly reviewed in Section~\ref{sec:xychain}), in this case the Jordan-Wigner transform leads to an effective one-particle Hamiltonian given by a 2$\times$2-block operator. Including a random field leads to Anderson-type diagonal blocks, separate for positive and negative energies, coupled linearly in the anisotropy parameter by off-diagonal blocks.

Random block operators of this type are less well studied than the Anderson model and, due to a lack of monotonicity properties, known properties of the Anderson model do not directly extend to random block operators. However, random block operators have recently found considerable interest, see \cite{KMM2010, ESS2012, GM2013, ES2013}, with physical motivation also provided by the Bogoliubov-de Gennes equation in the mean-field approximation of BCS theory. We also mention that the model given by the anisotropic XY chain has been considered in the quantum information literature under the name {\it Majorana chain}, see \cite{Kitaev} and \cite{BravyiKoenig}. In particular, \cite{BravyiKoenig} considers the effect of disorder in the Majorana chain on quantum memory, assuming that the underlying one-particle Hamiltonian satisfies a strong multi-point dynamical localization condition (for which there are currently no rigorous proofs).

The papers \cite{ESS2012} and \cite{GM2013} have provided localization proofs for random block operators. The authors of \cite{GM2013}, building on \cite{KMM2010}, use a Wegner estimate and a Lifshitz tail bound and adapt the bootstrap multiscale analysis of Germinet and Klein to prove spectral and dynamical localization at internal band edges for a class of random block operators. In \cite{ESS2012} the fractional moment method is adapted to prove localization of the entire spectrum for a large class of random block operators (including those arising from the anisotropic XY chain) in the {\it large disorder} regime.

What drives our motivation here is that the random block operator arising from the XY spin chain is {\it one-dimensional} (the reduction to an effective one-particle Hamiltonian via Jordan-Wigner transform does not work for multi-dimensional XY systems). Thus, suggested by physics as well as past experience, one expects that the entire spectrum should be localized at arbitrary (in particular, small) disorder strength. A proof of this is our goal here.

We start in Section~\ref{sec:xychain} with a brief review of the relation of the XY chain with effective one-particle Hamiltonians in the form of 2$\times$2-block operators. After this we state Theorem~\ref{thm:main}, our main result, before outlining the contents of the rest of the paper at the end of Section~\ref{sec:xychain}.

The present work is based on the thesis \cite{Chap2013} by the first named author. We will frequently refer to \cite{Chap2013} for additional background and more details.

%\section*{Acknowledgements}

%%%%%%%%%%%%%%%%%%%%%%%%%%%%%%%%%%%%%%%%%%%%%%%%%%%%%%%%%%
\section{The XY chain, associated block operators, and the main result} \label{sec:xychain}

The anisotropic XY spin chain, as introduced in \cite{LSM1961} for the constant coefficient case, is given by the self-adjoint Hamiltonian
\begin{equation}\label{xychain1}
H_n=\sum_{j=1}^{n-1} \mu_j[(1+\gamma_j)\sigma_j^x\sigma_{j+1}^x+(1-\gamma_j)\sigma_j^y\sigma_{j+1}^y]+\sum_{j=1}^n \nu_j \sigma_j^z,
\end{equation}
which acts on the Hilbert space $\mathcal{H}_n=\bigotimes_{j=1}^n \C^2$.
Here $\set{\mu_j}$, $\set{\gamma_j}$, and $\set{\nu_j}$ are three real-valued sequences, representing the coupling strength, the anisotropy, and the external magnetic field, respectively, and
\begin{equation}  \label{eq:Pauli}
\sigma^x=\left(
\begin{array}{cc}
0 & 1 \\
1 & 0 \\
\end{array}
\right),\quad\quad
\sigma^y=\left(
\begin{array}{cc}
0 & -i \\
i & 0 \\
\end{array}
\right),\quad\quad
\sigma^z=\left(
\begin{array}{cc}
1 & 0 \\
0 & -1 \\
\end{array}
\right)
\end{equation}
are the Pauli matrices. For a $2\times2$ matrix $M$, we use the notation $M_j:=I\otimes\cdots\otimes I\otimes M\otimes I\otimes\cdots\otimes I$, which acts non-trivially only in the $j$th component of $\mathcal{H}_n$. 

We remark that the first sum in (\ref{xychain1}) models the interactions between neighboring spins, as governed by the $x$ and $y$ Pauli matrices. The second sum models an exterior magnetic field acting on the spin system. The {\it anisotropic} case is characterized by $\gamma_j \not= 0$.

Define the raising and lowering operators
\begin{equation} 
a_j^*=\frac12(\sigma_j^x+i\sigma_j^y),\quad\quad a_j=\frac12(\sigma_j^x-i\sigma_j^y),\quad\quad j=1,...,n,
\end{equation}
and the Jordan-Wigner transform
\begin{equation}\label{JWT}
c_1:=a_1,\quad\quad c_j:=\sigma_1^z\cdots \sigma_{j-1}^z a_j,\quad\quad j=2,...,n.
\end{equation}
It turns out that the latter satisfy the canonical anticommutation relations (CAR)
\begin{equation}\label{CAR}
\{c_j,c_k^*\}=\delta_{jk}I,\quad\quad \{c_j,c_k\}=\{c_j^*,c_k^*\}=0, \quad\quad 1\le j,k\le n.
\end{equation}
Referring to \cite{HSS2012} for details and defining the formal vector
\begin{equation} 
\mathcal{C}=(c_1,...,c_n,c_1^*,...,c_n^*)^t,
\end{equation}
one is able to rewrite $H_n$ as
\begin{equation}\label{CMC}
H_n=\mathcal{C}^*\hat M_n\mathcal{C},
\end{equation}
where $\hat M_n$ is the $2\times2$ block matrix
\begin{equation}\label{blockmatrixM}
\hat M_n=\left(
\begin{array}{cc}
A_n & B_n \\
-B_n & -A_n \\
\end{array}
\right)
\end{equation}
with Jacobi matrices
\begin{equation}\label{matrixA}
A_n=\left(
\begin{array}{ccccc}
\nu_1 & -\mu_1 &  &  &  \\
-\mu_1 & \ddots & \ddots &  &  \\
 & \ddots & \ddots & \ddots &  \\
 &  & \ddots & \ddots & -\mu_{n-1} \\
 &  &  & -\mu_{n-1} & \nu_n \\
\end{array}
\right)
\end{equation}
and
\begin{equation}\label{matrixB}
B_n=\left(
\begin{array}{ccccc}
0 & -\mu_1\gamma_1 &  &  &  \\
\mu_1\gamma_1 & \ddots & \ddots &  &  \\
 & \ddots & \ddots & \ddots &  \\
 &  & \ddots & \ddots & -\mu_{n-1}\gamma_{n-1} \\
 &  &  & \mu_{n-1}\gamma_{n-1} & 0 \\
\end{array}
\right).
\end{equation}

The spin Hamiltonian $H_n$ is quadratic in the fermionic creation and annihilations operators $c_j^*$, $c_k$, $1\le j,k \le n$. However, due to anisotropy and the resulting off-diagonal blocks in (\ref{blockmatrixM}), it is not {\it particle number preserving}, i.e.\ the terms in (\ref{CMC}) do not contain equal numbers of $c_j^*$'s and $c_k$'s.

Observe that $A_n^*=A_n^t=A_n$ and $B_n^*=B_n^t=-B_n$ so that we have $\hat M_n^*=\hat M_n^t=\hat M_n$. The self-adjoint block matrix $\hat{M}_n$ can be considered an effective one-particle Hamiltonian for the many-body spin Hamiltonian $H_n$. In particular, as observed in \cite{HSS2012} and further discussed in Section~\ref{sec:discussion} below, one-particle localization properties of $\hat{M}_n$ imply many-body localization for $H_n$.

We will generally rewrite $\hat{M}_n$ as a Jacobi matrix with $2\times 2$-matrix-valued entries (formally obtained by the unitary equivalence under re-ordering the canonical basis $(e_1,e_2,...,e_{2n})$ as $(e_1,e_{n+1},e_2,e_{n+2},...,e_n,e_{2n})$),
\begin{equation}\label{blockmatrixhatM}
M_n := \left(
\begin{array}{cccc}
\nu_1 \sigma^z & -\mu_1 S(\gamma_1) &  &   \\
-\mu_1 S(\gamma_1)^t & \nu_2 \sigma^z & \ddots  &  \\
  & \ddots & \ddots & -\mu_{n-1}S(\gamma_{n-1}) \\
  &  & -\mu_{n-1}S(\gamma_{n-1})^t & \nu_n \sigma^z \\
\end{array}
\right),
\end{equation}
with the Pauli matrix $\sigma^z$ from (\ref{eq:Pauli}) and
\begin{equation} \label{defSgamma}
S(\gamma):=\left(
\begin{array}{cc}
1 & \gamma \\
-\gamma & -1 \\
\end{array}
\right).
\end{equation}

In the following we will think of $M_n$ as an operator on $\ell^2([1,n];\C^2)$, where $[1,n]$ denotes the discrete interval $\{1,\ldots,n\}$.
We may view $M_n$ as a generalized tight-binding Hamiltonian with a (sign-indefinite) potential $\nu_j \sigma^z$ and non-standard hopping terms $-\mu_j S(\gamma_j)$. Mathematically, (\ref{blockmatrixhatM}) provides the possibility to investigate spectral properties of $\hat M_n$ by using a transfer matrix formalism, although of higher order than in the case of standard tri-diagonal Jacobi matrices.

To understand the effects of disorder on spin systems, we are particularly interested in the case where at least one of the sequences $\set{\mu_j}$, $\set{\gamma_j}$, and $\set{\nu_j}$ is random. As the main application of our results below we choose the case of random exterior field (but see Section~\ref{sec:otherrand} for other cases). More precisely, let
\begin{equation} \label{eq:coeffcond1}
\mu_j = 1, \quad \gamma_j = \gamma \in (0,1) \cup (1,\infty) \quad \mbox{for all $j\in \N$},
\end{equation}
and
\begin{equation} \label{eq:coeffcond2}
\begin{array}{c} (\nu_j)_{j\in\N} \:\:\mbox{are i.i.d.\ real random variables with non-trivial,}\\ \mbox{ compactly supported distribution $\rho$}. \end{array}
\end{equation}

Physically most interesting is the case of anisotropy parameter $\gamma \in (0,1)$, but our methods also work for $\gamma>1$. We could also allow negative $\gamma$, which merely changes the roles of the Pauli matrices $\sigma^x$ and $\sigma^y$ in (\ref{xychain1}). For the isotropic XY chain $\gamma=0$, the off-diagonal blocks in (\ref{blockmatrixM}) vanish, reducing $\hat{M}_n$ to the Anderson model. This case was considered in \cite{HSS2012}, where known facts on dynamical localization for the Anderson model were used to establish dynamical localization for the isotropic XY chain in random field, see Section~\ref{sec:discussion} below. For $\gamma=1$ the XY chain (\ref{xychain1}) degenerates into the quantum Ising model. This can also be studied with our methods. In fact, in this case the structure of the underlying effective one-particle operator is simpler, allowing to obtain more complete results, which will be the content of \cite{IsingPaper} (see also more discussion at the end of Section~\ref{sec:discussion} below).

Let $P_j:\ell^2([1,n];\C^2)\to\C^2$ be the projection defined by $P_ju=u(j)$, and denote by $\chi_J(M_n)$ the spectral projection corresponding to $M_n$ onto $J\subset \R$. The following result on dynamical localization for $M_n$ may be considered as the main result of our paper.

\begin{thm} \label{thm:main}
Assume that $M_n$ is given by (\ref{blockmatrixhatM}) with coefficients as in (\ref{eq:coeffcond1}) and (\ref{eq:coeffcond2}).
For every compact interval $J\subset \R \setminus \{0\}$ and every $\zeta\in(0,1)$ there exist constants $C=C(J,\zeta)<\infty$ and $\eta = \eta(J,\zeta)>0$ such that for all $n\in\N$ and $j,k\in[1,n]$,
\begin{equation}\label{DL2}
\E\left(\sup_{t\in\R}\|P_je^{-itM_n} \chi_J(M_n) P_k^*\|\right)\le Ce^{-\eta\abs{j-k}^\zeta}.
\end{equation}
\end{thm}

Before proceeding we have several remarks. First, while we have chosen to state (\ref{DL2}) as a result on dynamical localization, what will really be proven is the stronger bound
\begin{equation}\label{DLgen}
\E\left(\sup_{|g|\le 1}\|P_j g(M_n) \chi_J(M_n) P_k^*\|\right)\le Ce^{-\eta\abs{j-k}^\zeta},
\end{equation}
where $g$ is an arbitrary function whose modulus is pointwise bounded by $1$ (in the finite volume case considered here, one does not even need to assume measurabilty of $g$).

Second, the bounds (\ref{DL2}) and (\ref{DLgen}) (for Borel functions $g$) hold under the same assumptions also for the infinite volume operators $M_{\nu,\gamma}$ introduced in Section~\ref{sec:infvol} below. One way to see this is by a limiting argument, using that the constants in the bound (\ref{DLgen}) do not depend on the size $n$ of the finite system (see, e.g., the argument in Section~6 of \cite{Stolz2011}). It further follows that the infinite volume operator almost surely has pure point spectrum (see the discussion of the hierarchy of localization properties in \cite{Klein2008}). Here it is not a problem that a vicinity of $E=0$ has to be excluded in the proof of dynamical localization, as one can exhaust $\R \setminus \{0\}$ by countably many compact intervals and the point $E=0$ alone cannot carry any continuous spectrum.

What is less clear, however, is if the exclusion of zero energy in Theorem~\ref{thm:main}, i.e.\ the need for the projection $\chi_J(M_n)$ in (\ref{DL2}) and (\ref{DLgen}), can be removed. This is due to a singularity of the transfer matrices at $E=0$. While the Lyapunov exponents will typically stay positive, the induced dynamical system loses its irreducibility. The exclusion of $E=0$ in the dynamical localization bound for the single-particle Hamiltonian is the reason that we currently cannot deduce many-body dynamical localization for the general anisotropic XY chain from Theorem~\ref{thm:main}, similar to what was done in \cite{HSS2012} for the isotropic chain. However, if the magnetic field $(\nu_j)$ is sufficiently strong, then it will be easily seen that $M_n$ has a spectral gap at $E=0$, so that in this special case Theorem~\ref{thm:main} will indeed lead to many-body dynamical localization for the anisotropic XY chain. We discuss this in more detail in our concluding Section~\ref{sec:discussion}, see in particular Theorem~\ref{thm:dynlocanisoXY}.

The other remaining sections are organized as follows: 

In Section~\ref{sec:infvol} we discuss some basic properties of the infinite volume operator corresponding to $M_n$. While this is not needed for the proof of Theorem~\ref{thm:main}, it gives some insight into the consequences of indefiniteness of the diagonal terms in (\ref{blockmatrixhatM}), in particular on the structure of the almost sure spectrum.

Sections~\ref{sec:Thouless} to \ref{sec:app} contain the proof of Theorem~\ref{thm:main}. We follow a strategy initially developed for the Anderson model on the strip in \cite{KLS1990}, meaning that the proof of localization is essentially reduced to showing Wegner and initial length estimates, from which bounds such as (\ref{DLgen}) follow by the bootstrap multiscale analysis of \cite{GeKle2001}. An advantage of this approach is that it allows to handle singular single-site distributions as in (\ref{eq:coeffcond2}). However, this comes at the price of having to choose $\zeta<1$ in (\ref{DL2}) and (\ref{DLgen}). In this context we should mention the work \cite{ESS2012}, where dynamical localization for a class of random block operators is proven, which includes our model, but requires large disorder and smooth single-site distributions. This work uses an adaptation of the fractional moment (or Aizenman-Molchanov) method and works for $\zeta=1$.

A key ingredient to deriving Wegner and initial length estimates is the Thouless formula, which we prove in Section~\ref{sec:Thouless} for block Jacobi matrices with general ergodic diagonal and off-diagonal terms. Previously, the Thouless formula for block Jacobi matrices has been shown only for the case of the Anderson model on a strip, where the off-diagonal blocks are identity matrices (see Section~\ref{sec:Thouless} for more discussion of earlier works).

In Section~\ref{sec:DL} we state Theorem~\ref{BJMdynloc}, a result on dynamical localization for block Jacobi matrices with i.i.d.\ diagonal and off-diagonal blocks, under the assumption that the F\"urstenberg groups associated with the transfer matrices have suitable contraction and irreducibility properties. We discuss its proof, which is patterned after the approach of \cite{KLS1990} for the Anderson model on strips. 

The proof of Theorem~\ref{thm:main} is completed in Section~\ref{sec:app} by showing that for the model (\ref{blockmatrixhatM}) the F\"urstenberg group has the properties required in Theorem~\ref{BJMdynloc}, using a criterion of Gol'dsheid and Margulis on Zariski-denseness of F\"urstenberg groups in the symplectic group. Here we also show, in Section~\ref{subsec:crit}, that at zero-energy the F\"urstenberg group is not Zariski-dense, while typically Lyapunov exponents are still positive. Without stating explicit results, we discuss in Section~\ref{sec:otherrand} what our methods yield if either the coupling constants $\mu_j$ or the anisotropy parameters $\gamma_j$ in (\ref{blockmatrixhatM}) are chosen random while keeping the other two parameter sets constant.

\section{The infinite volume operator and its basic properties} \label{sec:infvol}

For applications to the {\it finite} XY chain (\ref{xychain1}) we need to look at the finite volume operators $\hat M_n$ in (\ref{blockmatrixM}) (for arbitrary $n$). But in the study of the latter it is natural to also consider the (bounded, self-adjoint) infinite volume operator \begin{equation}\label{blockoperatorH}
\hat {M}_{\nu,\gamma}=\left(
\begin{array}{cc}
A_\nu & \gamma B \\
-\gamma B & -A_\nu \\
\end{array}
\right)
\end{equation}
on $\ell^2(\Z)\oplus\ell^2(\Z)$. Here $\nu = (\nu_j)_{j\in\Z}$ are i.i.d.\ random variables as in (\ref{eq:coeffcond2}) and the operators $A_\nu$ and $B$ on $\ell^2(\Z)$ are defined by
\begin{eqnarray}
(A_\nu u)(n) &=& -u(n+1)-u(n-1)+\nu_n u(n) \\
(Bu)(n) &=& u(n-1)-u(n+1) \nonumber
\end{eqnarray}
for all $n\in \Z$. This is unitarily equivalent to the random block Jacobi matrix
\begin{equation}\label{blockJacobi}
{M}_{\nu,\gamma}:= \left(
\begin{array}{ccccc}
\ddots & \ddots &  &  &  \\
\ddots & \nu_{-1} \sigma^z & -S(\gamma) &  &  \\
 & -S(\gamma)^t & \nu_0 \sigma^z & -S(\gamma) &  \\
 &  & -S(\gamma)^t & \nu_1 \sigma^z & \ddots \\
 &  &  & \ddots & \ddots \\
\end{array}
\right)
\end{equation}
in $\ell^2(\Z;\C^2)$.

% Symmetry of spectrum, periodic support theorem, spectrum not necessarily given by constant potentials

We begin by mentioning some basic spectral properties of general block operators of the form (\ref{blockoperatorH}). For this let $\mathcal{H}$ be a Hilbert space, and let $A\in B(\mathcal{H})$ be self-adjoint and $B\in B(\mathcal{H})$ be skew-adjoint (i.e. $B^*=-B$). Then the block operator
\begin{equation} \label{generalblock}
\hat{M}:=\left(
\begin{array}{cc}
A & B \\
-B & -A \\
\end{array}
\right)
\end{equation}
is bounded and self-adjoint on the Hilbert space $\mathcal{H}\oplus\mathcal{H}$. It is easy to see that $\|\hat{M}\| \le \|A\| + \|B\|$ and thus $\sigma(\hat{M})\subset[-\|A\|-\|B\|,\|A\|+\|B\|]$. Also, with the unitary
\begin{equation}
U:= \left( \begin{array}{cc} 0 & I \\ I & 0 \end{array} \right)
\end{equation}
one has $U^* \hat{M} U = -\hat{M}$, so that $\sigma(\hat{M}) = -\sigma(\hat{M})$. Finally, we have

\begin{prop}\label{specsymlemma} 
Let $\hat{M}$ be given by (\ref{generalblock}). If there exists a $\lambda >0$ such that $A\ge \lambda$ or $-A\ge\lambda$, then
\begin{equation} 
\sigma(\hat{M})\cap(-\lambda,\lambda)=\emptyset.
\end{equation}
\end{prop}

This is proven in Chapter 3 of \cite{Chap2013} and adapted from a result in \cite{KMM2010} for a similar class of block operators.

Returning to the random block operators $\hat{M}_{\nu,\gamma}$ and the unitarily equivalent random block Jacobi matrices $M_{\nu,\gamma}$, Proposition~\ref{specsymlemma} and the remark preceding it establish symmetry of the spectrum about 0 and some basic spectral inclusions for {\it every} choice of the i.i.d.\ random parameters $\nu = (\nu_j)$. Moreover, as for every $\gamma \in \R$ the family $\{{M}_{\nu,\gamma}\}_{\nu= (\nu_j)}$ is ergodic with respect to shifts in $\ell^2(\Z;\C^2)$ (for a discussion of general ergodic block Jacobi matrices see Section~\ref{sec:Thouless} below), there exists a closed subset $\Sigma_\gamma$ of $\R$, called the {\it almost sure spectrum}, such that 
\begin{equation} 
\Sigma_\gamma=\sigma({M}_{\nu,\gamma}) \quad \mbox{for almost every $\nu$}.
\end{equation}

For the Anderson model $A_{\nu}$ the almost sure spectrum is explicitly given by $[-2,2]+\supp \rho$, e.g.\ \cite{CL1990}. This can be understood as saying that the almost sure spectrum is generated as the union of all spectra where the potential takes a constant value in supp$\,\rho$. 

The non-monotonicity of the block Jacobi matrix $M_{\nu,\gamma}$ in the random parameters $\nu_j$ makes the description of the almost sure spectrum $\Sigma_{\gamma}$ for $\gamma \not= 0$ more complicated. By extending well-known arguments (e.g.\ \cite{Kir1989}) it is not hard to show that $M_{\nu,\gamma}$ satisfies a {\it periodic support theorem}. For this, denote
\begin{equation} 
S_{\text{per}} := \{V:\Z\to\R: V \text{ periodic}, V(n)\in\text{supp}\,\rho \ \text{for all} \ n\in\Z\}.
\end{equation}
\begin{thm}[Periodic Support Theorem]\label{PST}
We have
\begin{equation}
\Sigma_\gamma=\overline{\bigcup_{V\in S_{\text{\rm per}}}\sigma({M}_{V,\gamma})}.
\end{equation}
\end{thm}

For a proof, using well-known arguments, we refer to Chapter 4 of \cite{Chap2013}. It is not generally true that constant potentials, i.e.\ the operators $M_{c,\gamma}$ where $c$ is a constant in supp$\,\rho$, suffice to generate the entire almost sure spectrum. However, a positive result in this direction is the following.

\begin{thm}\label{equalsunionoverconstant}
If $\gamma\in[0,1]$, and $\supp\,\rho=[a,b]$ with $2\le a<b$, then
\begin{equation} 
\Sigma_\gamma=\bigcup_{c\in\supp\,\rho}\sigma({M}_{c,\gamma}).
\end{equation}
\end{thm}

We again refer to Chapter 4 of \cite{Chap2013} for a full proof, where an important ingredient is that the assumption $2\le a < b$ means that the block operator is in the gapped case of Proposition~\ref{specsymlemma}, meaning that the spectra of the diagonal blocks do not overlap.  If we remove the assumption $a\ge 2$, we can find examples where periodic potentials generate more spectrum than just the constant potentials. For example, if $\gamma=1/2$ and $\supp\,\rho=[-1,1]$, it is shown in \cite{Chap2013} that
\begin{equation} \label{eq:constgap}
\bigcup_{c\in\text{supp}\,\rho}\sigma({M}_{c,1/2})
=\left[-3,-\sqrt{2/3}\right] \cup\left[\sqrt{2/3},3\right] \subsetneq [-3,3]=\Sigma_{1/2}.
\end{equation}
In fact, the periodic potential $V=(...,-1,1,-1,1,...)$ fills the spectral gap $(-\sqrt{2/3}, \sqrt{2/3})$ left in (\ref{eq:constgap}), as $\sigma({M}_{V,1/2})=[-\sqrt5,\sqrt5]$. Examples such as these lead to the question of whether one can prove that 2-periodic potentials are always enough, and if this characterizes the almost sure spectrum in cases other than the one covered in Theorem~\ref{equalsunionoverconstant}.

%%%%%%%%%%%%%%%%%%%%%%%%%%%%%%%%%%%%%%%%%%%%%%%%%%%%%%%%%%
\section{A Thouless formula} \label{sec:Thouless}

As explained at the end of Section~\ref{sec:xychain}, we prove dynamical localization for the random block Jacobi matrices (\ref{blockmatrixhatM}) by adapting the strategy used in \cite{KLS1990} to prove localization for the Anderson model on a strip. A core part of this strategy, allowing to deduce regularity of the integrated density of states from regularity of the Lyapunov exponents, is the Thouless formula. 

For the strip case, two different proofs of the Thouless formula are in the literature, the original one by Craig and Simon in \cite{CS1983}, and a later proof by Kotani and Simon in \cite{KoSi1988}, where the Thouless formula arises quite naturally out of an extension of large parts of Kotani theory to the strip. In both of these works the off-diagonal blocks are chosen as identity operators, as is the case for the Anderson model on a strip. Here we need to discuss how to extend the Thouless formula to more general off-diagonal blocks. While a case could be made for following Kotani-Simon and extending their work to general block Jacobi matrices, we will follow the original approach of \cite{CS1983}.

We do this for general ergodic block Jacobi matrices. Let $(\Omega,\mathcal{F},\PP)$ be a complete probability space, $\ell \in \N$, and let $f,g:\Omega\to\R^{\ell\times\ell}$ be measurable such that for a.e.\ $\omega\in\Omega$, $f(\omega)$ is symmetric and $g(\omega)$ is invertible. Let us further assume that there exists $D\in (0,\infty)$ such that
\begin{equation}\label{boundedassumption2}
\norm{f(\omega)}+\norm{g(\omega)}+\|g(\omega)^{-1}\|\le D\quad\text{almost surely.}
\end{equation}
Let $T:\Omega\to\Omega$ be an ergodic bijection. Then we define the random operator
\begin{equation}\label{BJM}
{M} = M(\omega) =\left(\begin{array}{ccccc}
\ddots & \ddots & & & \\
\ddots & V_{-1} & -S_{-1} & & \\
& -S_{-1}^t & V_0 & -S_0 & \\
& & -S_0^t & V_1 & \ddots \\
& & & \ddots & \ddots 
\end{array}\right)
\end{equation}
on $\ell^2(\Z;\C^\ell)$, where 
\begin{equation} \label{eq:defVS}
V_n(\omega):=f(T^n\omega) \quad \text{and} \quad S_n(\omega):=g(T^n\omega).
\end{equation} 
In this way, we have
\begin{equation}\label{ergprop}
{M}(T\omega)=U{M}(\omega)U^*,
\end{equation}
where
\begin{equation} 
(U\varphi)(j)=\varphi(j+1),\quad\varphi\in\ell^2(\Z;\C^\ell)
\end{equation}
is the left-shift operator. The family $\set{{M}(\omega)}_{\omega\in\Omega}$ is what we call an {\it ergodic block Jacobi matrix}. Note that by construction and (\ref{boundedassumption2}), ${M}$ is almost surely bounded and self-adjoint on $\ell^2(\Z;\C^\ell)$.

Define the finite volume operator
\begin{equation}\label{Hn1n2}
{M}_{[n_1,n_2]}:=\left( \begin{array}{cccc} V_{n_1} & -S_{n_1} & & \\ -S_{n_1}^t & V_{n_1+1} & \ddots & \\ & \ddots & \ddots & -S_{n_2-1} \\ & & -S_{n_2-1}^t & V_{n_2} \end{array} \right),
\end{equation}
which is the restriction of ${M}$ to $\ell^2([n_1,n_2];\C^\ell)$. For shorthand, we write ${M}_n:={M}_{[1,n]}$.

Define the projection $P_0:\ell^2(\Z;\C^\ell)\to\C^\ell$ by $P_0u=u(0)$. Then on Borel sets $A\subset \R$ the {\it density of states measure}
\begin{equation} \label{eq:DOS}
dN(A)= \E(\tr(P_0 \chi_A(M)P_0^*)) = \lim_{n\to\infty}\frac{1}{\ell n}\tr(\chi_A({M}_n)) \qquad\text{almost surely}
\end{equation}
exists by the standard arguments (e.g.\ Chapter 5 of \cite{Kir2007}), which hold also in this general ergodic setting. In fact, ignoring the limit, the quantity on the right-hand side defines an integrated density of states measure for ${M}_n$ which converges weakly to $dN$. We define the integrated density of states (IDS) to be the distribution function $E\mapsto N(E)$ of the measure $dN$. Its set of growth points is the almost sure spectrum $\Sigma$ of $M$.

We introduce the (modified) $2\ell \times 2\ell$ transfer matrices
\begin{equation}\label{tmprop}
A_k^E = \left(\begin{array}{cc}
0 & S_{k-1}^{-1} \\
-S_{k-1}^t & (V_k-E)S_{k-1}^{-1} \\
\end{array}\right), \quad E\in \C, \quad k=1,...,n
\end{equation}
and the $k$-step transfer matrices $T_k^E:=A_k^E\cdots A_1^E$, $k=1,...,n$. 
One can check that the properties
\begin{equation} 
{u(k)\choose S_ku(k+1)} = A_k^E{u(k-1)\choose S_{k-1}u(k)}, \quad\quad k=1,...,n
\end{equation}
are equivalent to $u: [0,n+1] \to \C^{\ell}$ solving the difference equation
\begin{eqnarray}
-S_{k-1}^tu(k-1)+V_ku(k)-S_ku(k+1) &=& Eu(k),\quad\quad k=1,...,n.\quad\quad \label{diffeq}
\end{eqnarray}

As our transfer matrices $A_k$ are symplectic, i.e.\
\begin{equation} 
A_k^tJA_k=J\quad\text{where}\quad J=\left(\begin{array}{cc}0 & I \\ -I & 0\end{array}\right),
\end{equation}
there exist $2\ell$ Lyapunov exponents, defined inductively by
\begin{equation} \label{eq:defLyap}
\gamma_1(E)+\cdots+\gamma_p(E)=\lim_{n\to\infty}\frac1n \log\|\wedge^p T_n^E\| \quad\text{a.s.,}\quad p=1,...,2\ell;
\end{equation}
see for example \cite{CL1990} for the definition of the exterior powers $\wedge^p T_n^E$. The Lyapunov exponents come in symmetric pairs about $0$,
\begin{equation} 
\gamma_1(E)\ge\cdots\ge\gamma_\ell(E)\ge 0 \ge \gamma_{\ell+1}(E)=-\gamma_\ell(E)\ge\cdots\ge \gamma_{2\ell}(E) = -\gamma_1(E).
\end{equation}
For existence in the ergodic setting, we refer to Theorem IV.2.6 in \cite{CL1990}, which holds also for complex energy.

The Thouless formula relates the IDS to the sum of the first $\ell$ Lyapunov exponents or, to stay with the convention in \cite{CS1983}, the Lyapunov index
\begin{equation} 
\gamma(E):=\frac1\ell[\gamma_1(E)+\cdots+\gamma_\ell(E)].
\end{equation}

\begin{thm}[Thouless Formula]\label{ThoulessThm}
Let $N(E)$ and $\gamma(E)$ be the IDS and Lyapunov index for an ergodic block Jacobi matrix as given by (\ref{boundedassumption2}) to (\ref{eq:defVS}). Then, for all $E\in \C$,
\begin{equation}\label{Thouless}
\gamma(E)=-\frac1\ell\E(\log|\det g|)+\int_\R \log\abs{E-E'}dN(E').
\end{equation}
\end{thm}

Recall from (\ref{eq:defVS}) that $S_n(\omega) = g(T^n \omega)$ and thus, due to ergodicity, $\E(\log |\det S_n|) = \E(\log |\det g|)$ for all $n$. This shows how (\ref{Thouless}) generalizes the Thouless formula for the Anderson model on a strip, where $g=I$ and thus the first term on the right-hand side of (\ref{Thouless}) vanishes, reducing the Thouless formula to its familiar form.

That non-standard hopping, i.e.\ $g\not= I$, leads to a modification of the Thouless formula is well known for the case $\ell=1$, i.e.\ standard ergodic Jacobi matrices; see remarks on page 376 of \cite{CL1990} or page 274 of \cite{PF1992}. But for $\ell>1$ our result seems to be new, see, however, a discussion of the Thouless formula in the context of general ergodic block Jacobi matrices in Section~4 of \cite{SchuBa1}.

\begin{proof}[Proof of Theorem \ref{ThoulessThm}]
We will outline the major steps of the proof, which extends the arguments of \cite{CS1983} to the more general case considered here. For further details, we refer to Chapter~6 of \cite{Chap2013}.

In the following, for technical convenience, whenever working with fixed $n$, we set $S_0 = S_n = I$. This does not affect the limit defining $\gamma(E)$ as only one factor in the product of transfer matrices is changed, and this factor satisfies uniform norm bounds due to assumption (\ref{boundedassumption2}). It also does not affect the definition (\ref{eq:DOS}) of the density of states measure as the matrix $M_n$ only depends on $S_1, \ldots, S_{n-1}$.

\vspace{.3cm}

\underline{Step 1}: We claim that
\begin{equation}\label{thouless1}
\ip{e_{\ell+1}\wedge\cdots\wedge e_{2\ell}, \wedge^\ell T_n^E(e_{\ell+1}\wedge\cdots\wedge e_{2\ell})}=\frac{(-1)^{\ell n}}{\prod_{j=1}^{n-1}\det S_j} \prod_{m=1}^{\ell n}(E-E_m),
\end{equation}
where $\set{E_m}_{m=1}^{\ell n}$ are the eigenvalues of $M_n$, each repeated according to multiplicity.

To see this, note that we can write
\begin{equation} 
T_n^E = \left( \begin{array}{cc} Q^E(n) & P^E(n) \\ Q^E(n+1) & P^E(n+1) \end{array} \right),
\end{equation}
where $Q^E$ and $P^E$ are the unique matrix-valued solutions of
\begin{equation} 
-S_{k-1}^t X(k-1) + V_k X(k) - S_k X(k+1) = EX(k), \quad k=1,\ldots, n,
\end{equation}
satisfying $Q^E(0)=I$, $Q^E(1)=0$, $P^E(0)=0$, and $P^E(1)=I$.

By properties of exterior products, the matrix element in (\ref{thouless1}) is exactly $\det P^E(n+1)$, and Step 1 is complete after noticing that
\begin{equation} 
\det({M}_n-E)=\left(\prod_{j=1}^{n-1}\det S_j\right)\det P^E(n+1),
\end{equation}
which may be proven by induction by performing appropriate elementary column transformations on ${M}_n-E$.

\vspace{.3cm}

\underline{Step 2}: For $E\in\C\setminus\R$, we have the first necessary inequality:
\begin{equation}\label{thouless2}
\gamma(E)\ge -\frac1\ell\E(\log\abs{\det g})+\int_\R \log\abs{E-E'}dN(E').
\end{equation}
This inequality follows immediately from Step 1 by estimating $\|\wedge^\ell T_n^E\|$ below by its matrix element given in (\ref{thouless1}), and using Birkhoff's ergodic theorem and that $dN_n \xrightarrow{w} dN$, where $dN_n=\frac{1}{\ell n}\sum_{m=1}^{\ell n}\delta_{E_m}$ is the density of states measure of ${M}_n$.

\vspace{.3cm}

\underline{Step 3}: We have
\begin{itemize}
\item[(a)] The set
\begin{equation} \label{totalset'}
\left\{ \left( \begin{array}{c} e_1 \\ M e_1 \end{array} \right) \wedge \cdots \wedge \left( \begin{array}{c} e_{\ell} \\ M e_{\ell} \end{array} \right): \; M\in\R^{\ell\times\ell}\right\}
\end{equation}
is total in $\wedge^{\ell} \C^{2\ell}$.

\item[(b)] The set
\begin{equation} \label{totalset'2}
\left\{ \left( \begin{array}{c} -e_1 \\ M^t e_1 \end{array} \right) \wedge \cdots \wedge \left( \begin{array}{c} -e_{\ell} \\ M^t e_{\ell} \end{array} \right): \; M\in\R^{\ell\times\ell}\right\}
\end{equation}
is total in $\wedge^{\ell} \C^{2\ell}$.
\end{itemize}

This is proven by induction, which, while tedious, uses little more than multilinearity of exterior products.

\vspace{.3cm}

\indent\underline{Step 4}: Let $d_\ell=\dim(\wedge^\ell\C^{2\ell})={2\ell\choose\ell}$. Then
\begin{itemize}
\item[(a)] There exist matrices $M_{-,j}\in\R^{\ell\times\ell}$, $j=1,...,2d_\ell$, such that
\begin{equation}\label{Step4a}
\left\{ \left( \begin{array}{c} e_1 \\ (M_{-,j}-E) e_1 \end{array} \right) \wedge \cdots \wedge \left( \begin{array}{c} e_{\ell} \\ (M_{-,j}-E) e_{\ell} \end{array} \right): j = 1,...,2d_\ell \right\}
\end{equation}
is total in $\wedge^{\ell} \C^{2\ell}$ for all $E\in \C$. 

\item[(b)] There exist matrices $M_{+,k}\in\R^{\ell\times\ell}$, $k=1,...,2d_\ell$, such that
\begin{equation}\label{Step4b}
\left\{ \left( \begin{array}{c} -e_1 \\ (M_{+,k}^t-\conj{E})e_1 \end{array} \right) \wedge \cdots \wedge \left( \begin{array}{c} -e_{\ell} \\ (M_{+,k}^t-\conj{E})e_{\ell} \end{array} \right): k = 1,...,2d_\ell \right\}
\end{equation}
is total in $\wedge^{\ell} \C^{2\ell}$ for all $E\in \C$.
\end{itemize}

Step 3 gives us matrices $M_1,...,M_{d_\ell}$ that yield (\ref{totalset'}) and (\ref{totalset'2}). The price we pay now in Step 4 for wanting such spanning properties for {\it all} $E\in\C$ is that we allow twice as many matrices as before. Notice that if we take $M_{-,j}=M_j$, $j=1,...,d_\ell$, then there are finitely many $E$ such that the vectors in (\ref{Step4a}) do not span $\wedge^\ell\C^{2\ell}$, for if we form a matrix whose columns are the (coordinate representations of the) vectors in (\ref{Step4a}), $j=1,...,d_\ell$, the determinant of this matrix is a polynomial in $E$ with finitely many roots. We can then shift all of these matrices by $\lambda_0 I$, where $\lambda_0\in\R$ is large enough so that we avoid these roots, and we let $M_{-,j}$, $j=d_\ell+1,...,2d_\ell$, be these shifts.

\vspace{.3cm}

\indent\underline{Step 5}: For $j,k=1,...,2d_\ell$, define the extended block Jacobi matrices
\begin{equation} \label{eq:HLext}
{M}_{n,j,k} = \left( \begin{array}{cccccc} M_{-,j} & -I & & & & \\ -I & V_1 & -S_1 & & & \\ &  -S_1^t & \ddots & \ddots & & \\ & & \ddots & \ddots & -S_{n-1} & \\ & & & -S_{n-1}^t & V_n & -I \\ & & & & -I & M_{+,k} \end{array} \right).
\end{equation}
This extended operator ${M}_{n,j,k}$ shares the same transfer matrices $A_k$, $k=1,...,n$ as ${M}_n$ (given that we argued earlier that we may set $S_0=S_n=I$). However, it possesses an additional transfer matrix at each end. These are
\begin{equation}
A_{0,j}^E=\left(\begin{array}{cc} 0 & I \\ -I & M_{-,j}-E\end{array}\right)\quad\text{and}\quad A_{n+1,k}^E=\left(\begin{array}{cc} 0 & I \\ -I & M_{+,k}-E\end{array}\right),
\end{equation}
which satisfy
\begin{equation}
{u(0)\choose u(1)} = A_{0,j}^E{u(-1)\choose u(0)} \quad\text{and}\quad {u(n+1)\choose u(n+2)} = A_{n+1,k}^E{u(n)\choose u(n+1)}
\end{equation}

Similar to Step 1, we have
\begin{equation}\label{extendedroots}
\langle e_{\ell+1}\wedge\cdots\wedge e_{2\ell},\wedge^\ell (A_{n+1,k}^E T_n^E A_{0,j}^E)(e_{\ell+1}\wedge\cdots\wedge e_{2\ell})\rangle = \frac{(-1)^{\ell n}}{\prod_{j=1}^{n-1}\det S_j}\prod_{m=1}^{\ell(n+2)}(E-\tilde E_m),
\end{equation}
where $\{\tilde E_m\}_{m=1}^{\ell(n+2)}\subset\C$ are the eigenvalues of ${M}_{n,j,k}$, which are no longer necessarily real if $M_{-,j}$ or $M_{+,k}$ is not symmetric (meaning also that we now have to count according to {\it algebraic} multiplicity).

\vspace{.3cm}

\underline{Step 6}: If $V_1$ and $V_2$ are finite-dimensional inner product spaces with finite spanning sets $S_1$ and $S_2$, respectively, then there is a $c>0$ such that for all linear transformations $A:V_1\to V_2$,
\begin{equation}
c\norm{A} \le \sup_{\varphi\in S_2,\psi\in S_1} \abs{\ip{\varphi,A\psi}}.
\end{equation}
This is easy to see by equivalence of norms, since the right-hand side can be seen to be a norm.

\vspace{.3cm}

\underline{Step 7}: For $E\in\C\setminus\R$, we have the reverse inequality:
\begin{equation}\label{thouless4}
\gamma(E)\le -\frac1\ell\E(\log\abs{\det g})+\int_\R \log\abs{E-E'}dN(E').
\end{equation}

This is the more difficult inequality. We estimate $\|\wedge^\ell T_n^E\|$ via Step 6, using the spanning sets given in Step 4. One can then see how (\ref{extendedroots}) arises, which allows us to finish the proof using Step 5 similar to how Step 2 was done using Step 1.

\vspace{.3cm}

\underline{Step 8}: The Thouless formula holds for all $E\in\R$ as well. By the arguments in \cite{CS1983S}, the left- and right-hand sides of (\ref{Thouless}) are subharmonic functions. Since we have shown (\ref{Thouless}) for $E\in\C\setminus\R$, which has full measure in $\C$, equality extends to all of $\C$.
\end{proof}

\section{Dynamical Localization} \label{sec:DL}
% Include Wegner estimate and ISLE

While the Thouless formula in Section~\ref{sec:Thouless} only required ergodicity of the block Jacobi matrix (\ref{BJM}), in order to get localization properties we now consider the case of i.i.d.\ random entries. More precisely, let $\set{V_n}$ be i.i.d.\ with common distribution $\mu^1$ compactly supported in the real symmetric $\ell \times \ell$-matrices, and $\set{S_n}$ i.i.d.\ with common distribution $\mu^2$ compactly supported in the real invertible $\ell \times \ell$-matrices. We also assume that the $\set{V_n}$ and $\set{S_n}$ are independent from each other. The boundedness assumption (\ref{boundedassumption2}) then becomes
\begin{equation}\label{boundedassumption}
\norm{V_0}+\norm{S_0}+\|S_0 ^{-1}\|\le D < \infty \quad\text{almost surely.}
\end{equation}

Notice that ${M}_{\nu,\gamma}$ from (\ref{blockJacobi}) is covered by this model when we choose $\ell=2$, $V_n=\nu_n \sigma^z$, and $S_n=S(\gamma)$.

For this class of random block Jacobi matrices, we prove dynamical localization via the bootstrap multiscale analysis (MSA) of Germinet and Klein \cite{GeKle2001}, under suitable contractivity and irreducibility assumptions on the F\"urstenberg group. These assumptions may be checked, for example, by showing Zariski-denseness, a concept which we will discuss in more detail in Section~\ref{sec:app}. To apply the bootstrap MSA, it is sufficient (see, e.g., Klein's survey \cite{Klein2008}) to show an appropriate Wegner estimate and initial length scale estimate. To this end, we adapt the approach of Klein, Lacroix, and Speis \cite{KLS1990}, which proves such estimates for Anderson models on strips. Two important inputs into this argument are the Thouless formula (to prove regularity of the IDS) and a representation formula for the Green function (allowing to turn positivity of Lyapunov exponents into exponential Green function decay).

We will now discuss this in some more detail but heavily refer to earlier works.

Under the assumptions of this section the modified transfer matrices $A_n^E$ in (\ref{tmprop}) are i.i.d.\ with common distribution $\mu_E$ compactly supported in $\Sp_{\ell}(\R)$, the $2\ell \times 2\ell$-symplectic matrices. We define the F\"urstenberg group
\begin{equation}\label{Furstgroup}
G_{\mu_E}:=\overline{\ip{\supp \,\mu_E}}
\end{equation}
to be the smallest closed subgroup of $\Sp_{\ell}(\R)$ containing $\supp \mu_E$. 

For $\ell=1$, F\"urstenberg's theorem, e.g.\ \cite{CL1990}, says that non-compactness and strong irreducibility of $G_{\mu_E}$ imply positivity of the Lyapunov exponent at $E$. For a higher-order analogue of F\"urstenberg's theorem one has to require that the F\"urstenberg group is $p$-contractive and $L_p$-strongly irreducible for $p=1,\ldots,\ell$ (see \cite{BL1985} for definitions). Thus we will assume that
\begin{equation} \label{DefCI}
\begin{array}{c} G_{\mu_E} \ \text{is $p$-contracting and $L_p$-strongly irreducible for every} \\ \text{$p\in\set{1,...,\ell}$ and $E\in I$, where $I\subset \R$ is an open interval.} \end{array}
\end{equation}

By Proposition IV.3.4 of \cite{BL1985}, (\ref{DefCI}) implies $\gamma_1(E)>\cdots>\gamma_\ell(E)>0$ for all $E\in I$. But much more is true: 

\begin{thm}[Dynamical Localization]\label{BJMdynloc}
If (\ref{DefCI}) holds, then for every compact interval $J\subset I$ and every $\zeta\in(0,1)$, there exist constants $C<\infty$ and $\eta>0$ such that for every $L\in\N$ and $j,k \in \Lambda_L$,
\begin{equation}\label{dynloc}
\E\left(\sup_{t\in\R}\|P_j \chi_J({M}_{\Lambda_L})e^{-it{M}_{\Lambda_L}}P_k^*\|\right) \le Ce^{-\eta|j-k|^\zeta}.
\end{equation}
\end{thm}

Here, $P_j:\ell^2(\Z;\C^\ell)\to\C^\ell$ is the projection $P_j u=u(j)$, and $\Lambda_L:=[-L,L]$.

In Section~\ref{sec:app} below we will discuss applications of Theorem~\ref{BJMdynloc} by verifying in certain examples that assumption (\ref{DefCI}) holds. In particular, this will complete the proof of Theorem~\ref{thm:main} on dynamical localization for block operators associated with the XY chain in random exterior field.

In the remainder of this section, we sketch what goes into the proof of Theorem \ref{BJMdynloc}. We will not present every needed technical result for two reasons: A complete proof may be found in Chapter 7 of \cite{Chap2013}, and the approach is essentially the same as in \cite{KLS1990}, with some streamlining and slight modifications due to the non-standard hopping.

First, one uses (\ref{DefCI}) to conclude local H\"older continuity of the Lyapunov exponents on $I$ (see page 279 of \cite{CL1990}): For every $p\in\set{1,...,\ell}$ and every compact interval $J\subset I$, there exist constants $C<\infty$ and $\alpha>0$ such that
\begin{equation}\label{gammaHolder}
\abs{\gamma_p(E)-\gamma_p(E')}\le C\abs{E-E'}^\alpha \quad \mbox{for all $E,E' \in J$}.
\end{equation}
The Thouless formula, Theorem \ref{ThoulessThm}, along with properties of the Hilbert transform, is then used to transfer this H\"older continuity to the IDS:
For every compact interval $J\subset I$ there exist constants $C<\infty$ and $\alpha>0$ such that
\begin{equation}\label{NHolder}
\abs{N(E)-N(E')}\le C\abs{E-E'}^\alpha \quad \mbox{for all $E,E' \in J$}.
\end{equation}
The proof of (\ref{NHolder}) is essentially the same as that of Theorem A.1 in \cite{CKM1987}. This H\"older continuity of the IDS is one of the key ingredients in the proof of a Wegner estimate, Theorem \ref{Wegner} below.

To prove a Wegner estimate, which can be thought of as a statement about the size of a resolvent, one needs to be able to express Green functions in terms of solutions of the associated difference equation. For standard Jacobi matrices this is a well-known classical formula. A block form of this formula which holds for discrete Schr\"odinger operators on strips can be found, for example, in Proposition~III.5.6 of \cite{CL1990}. What we need here can be found in Section~2 of \cite{SchuBa2} and is a version of this formula which holds for general block Jacobi matrices, including non-standard hopping terms. To state it, we first define the (modified) Wronskian of two matrix-valued functions $U,V:[0,L+1]\to\C^{\ell\times\ell}$ as
\begin{equation}\label{Wronsk}
(W(U,V))(k) := (V(k))^t S_k U(k+1) - (S_k V(k+1))^t U(k), \quad k=0,...,L.
\end{equation}
This is shown to be constant (in $k$) if $U$ and $V$ are solutions of the finite difference equation
\begin{equation}\label{matrixFDE}
-S_{k-1}^tX(k-1)+V_kX(k)-S_kX(k+1)=zX(k),\quad k=0,...,L.
\end{equation}
If we let $U^z$ and $V^z$ be the unique solutions of (\ref{matrixFDE}) satisfying $U^z(0)=0$, $U^z(1)=I$, $V^z(L)=I$, and $V^z(L+1)=0$, then $W(U^z,V^z)$ is invertible if and only if $z\notin\sigma({M}_L)$. For such $z$ and for $j,k\in\set{1,...,L}$, we can also define the block Green function $G_L(j,k;z):=P_j({M}_L-z)^{-1}P_k^*$. For any $z\in\C\setminus\sigma({M}_L)$, we have the Green function formula
\begin{equation} \label{eq:Green}
G_L(j,k;z) = \left\{ \begin{array}{ll} U^z(j) W(U^z,V^z)^{-1} V^z(k)^t, & \mbox{if $j\le k$} \\ V^z(j) (W(U^z,V^z)^t)^{-1} U^z(k)^t, & \mbox{if $j\ge k$}. \end{array} \right.
\end{equation}
In \cite{SchuBa2} the above facts are discussed within a derivation of Weyl theory for general block Jacobi matrices. Also see Appendix~B of \cite{Chap2013} for a proof of (\ref{eq:Green}). 

Choosing $(j,k)=(1,L)$ in (\ref{eq:Green}) and $k=L$ in (\ref{Wronsk}) yields the relation
$G_L(1,L;z)=(S_LU^z(L+1))^{-1}$, which provides a crucial link between growth properties of solutions of (\ref{matrixFDE}), e.g.\ Lyapunov exponents, and decay of Green's function. In particular, this combined with H\"older continuity of the IDS (\ref{NHolder}) allows to adapt arguments of \cite{KLS1990} to prove the following Wegner estimate. Again, we refer to Chapter 7 of \cite{Chap2013} for a full proof.

\begin{thm}[Wegner Estimate]\label{Wegner}
Suppose (\ref{DefCI}) holds. For any $\beta\in(0,1)$, $\sigma>0$, and compact interval $J\subset I$, there exist $L_0=L_0(J,\beta,\sigma)$ and $\tau=\tau(J,\beta,\sigma)>0$ such that
\begin{equation} 
\PP(d(E,\sigma({M}_{\Lambda_L}))\le e^{-\sigma L^\beta}) \le e^{-\tau L^\beta}
\end{equation}
for all $E\in J$ and $L\ge L_0$.
\end{thm}

Using Theorem \ref{Wegner} and some geometric resolvent identities, one proves the following initial length scale estimate, similarly to how it is done in \cite{KLS1990}, but accounting for the non-standard hopping terms by absorbing the upper bound (\ref{boundedassumption}) into various constants.

\begin{thm}[Initial Length Scale Estimate]
Suppose (\ref{DefCI}) holds, and let $E_0\in I$. For every $\eps>0$ and $\beta\in(0,1)$ there exist $L_1=L_1(E_0,\eps,\beta)$ and $\kappa=\kappa(E_0,\eps,\beta)>0$ such that
\begin{equation} 
\PP\left(E_0\notin\sigma({M}_{\Lambda_L})\ \text{and}\ \|G_{\Lambda_L}(0,L;E_0)\| \le e^{-(\gamma_\ell(E_0)-\eps)L/16}\right) \ge 1-e^{-\kappa L^\beta}
\end{equation}
for all $L\ge L_1$.
\end{thm}

As remarked earlier, the proof of dynamical localization now follows using the bootstrap multiscale analysis of Germinet and Klein \cite{GeKle2001}, see also the survey \cite{Klein2008} which stresses that a wide range of localization properties follow for a broad class of models once input assumptions such as provided above have been shown. Strictly speaking, the works \cite{GeKle2001} and \cite{Klein2008} only discuss the infinite volume version of (\ref{dynloc}). The finite volume version stated above turns out to be equivalent by Theorem~6.1 of \cite{GeKlo2012}.

%%%%%%%%%%%%%%%%%%%%%%%%%%%%%%%%%%%%%%%%%%%%%%%%%%%%%%%%%%
\section{Applications} \label{sec:app}

\subsection{Checking contractivity and irreducibility} \label{subsec:checking}
% Background on Zariski-denseness

We showed in Theorem~\ref{BJMdynloc} that if a random block Jacobi matrix satisfies the assumption (\ref{DefCI}) on some open interval $I$, then it exhibits dynamical localization in the form (\ref{dynloc}). To apply Theorem~\ref{BJMdynloc} and, in particular, prove Theorem~\ref{thm:main}, we have to verify (\ref{DefCI}) in concrete examples. For this we will use the criterion of Zariski-denseness of the F\"urstenberg group, which Gol'dsheid and Margulis used in \cite{GoMa1989} to show positivity of Lyapunov exponents for the Anderson model on strips. More recently, this criterion has also been used in the proof of localization properties for continuum Anderson-type models with matrix-valued potential, see \cite{Bou2009} and references therein, and for a class of unitary random operators \cite{BM2013}.

Let $\R^{2\ell\times2\ell}$ denote the space of all $2\ell\times 2\ell$ real matrices, let us identify $\R^{2\ell\times2\ell}\cong\R^{(2\ell)^2}$, and let $\R[x_1,...,x_{(2\ell)^2}]$ denote the ring of all real-coefficient polynomials in the variables $x_1,...,x_{(2\ell)^2}$. The Zariski topology on $\R^{2\ell\times2\ell}$ is defined by declaring the following type of sets to be closed:
\begin{equation} 
V(S):=\{x\in\R^{(2\ell)^2}:\forall P\in S,\ P(x)=0\},\quad \mbox{for all $S\subset\R[x_1,...,x_{(2\ell)^2}].$}
\end{equation}
Such a set $V(S)$ is called an {\it algebraic variety} and is the set of common zeros of all polynomials from $S$. Then the Zariski topology on $\Sp_\ell(\R)$ is just the topology induced by the Zariski topology on $\R^{2\ell\times2\ell}$.

The {\it Zariski closure} $\Cl_Z(G)$ of a subset $G$ of $\Sp_\ell(\R)$ is the smallest closed set in the Zariski topology that contains $G$, i.e.\ if $G\subset\Sp_\ell(\R)$, then $\Cl_Z(G)$ is the set of zeros of polynomials vanishing on $G$. A subset $G'\subset G$ is said to be {\it Zariski-dense} in $G$ if $\Cl_Z(G')=\Cl_Z(G)$, i.e.\ each polynomial vanishing on $G'$ also vanishes on $G$.

\begin{thm}[Gol'dsheid-Margulis Criterion, \cite{GoMa1989}]\label{GMC}
Suppose $\set{B_n}_{n\in\N}\subset\Sp_\ell(\R)$ are i.i.d.\ random matrices with common distribution $\mu$, and $G_\mu:=\overline{\ip{\supp \mu}}$ is the F\"urstenberg group. If $G_\mu$ is Zariski-dense in $\Sp_\ell(\R)$, and thus $\Cl_Z(G_\mu)=\Sp_\ell(\R)$, then for every $p\in\set{1,...,\ell}$, $G_\mu$ is $p$-contracting and $L_p$-strongly irreducible.
\end{thm}

Showing that the Lie groups $\Cl_Z(G_\mu)$ and $\Sp_\ell(\R)$ are equal is equivalent to showing that the associated Lie algebras $\mathfrak{G}_\ell$ and $\mathfrak{sp}_\ell(\R)$ are equal. The latter has dimension $\ell(2\ell+1)$ on account of the characterization
\begin{equation}\label{sympLiealg}
\mathfrak{sp}_\ell(\R)=\set{\left(\begin{array}{cc} a & b_1 \\ b_2 & -a^t \end{array}\right):a\in \R^{\ell \times \ell} ;\ b_1,b_2\in \R^{\ell \times \ell} \text{ symmetric}}.
\end{equation}
Thus, one can conclude Zariski-denseness of $G_\mu$ in $\Sp_\ell(\R)$ if one can find $\ell(2\ell+1)$ linearly independent elements in $\mathfrak{G}_\ell$.

Applications of Theorem~\ref{BJMdynloc} to random block Jacobi matrices are found by verifying the Gol'dsheid-Margulis criterion for the F\"urstenberg groups $G_{\mu_E}$ from (\ref{Furstgroup}) for all $E\in I$. Naturally, for any given model, we will be interested in finding the largest open interval where this holds. For the Anderson model on a strip it was shown in \cite{GoMa1989} that one may choose $I=\R$, i.e.\ the model is dynamically localized at all energies. The authors of \cite{BS2006} consider an explicit example with $\ell=2$ (associated with a continuum Anderson-type model on two coupled strings), where they can verify Zariski-denseness for all $E$ outside a discrete set of critical energies. For the model considered in Theorem~\ref{thm:main} we can stay quite close to the arguments of \cite{BS2006} and see in the next subsection that one may choose $I= \R \setminus \{0\}$. We will also see that $E=0$ is indeed a critical energy at which the F\"urstenberg group is not Zariski-dense.

\subsection{Proof of Theorem~\ref{thm:main}} \label{subsec:anisotropic}
% Show Zariski-denseness at nonzero E, and non-Zariski- denseness at 0

We now return to the example of a random block Jacobi matrix considered in Theorem~\ref{thm:main}. As explained above, Theorem~\ref{thm:main} follows from Theorem~\ref{BJMdynloc} and

\begin{thm}\label{Zariskidenseneq0}
Let $M$ the random block Jacobi matrix of type (\ref{BJM}) with $\ell=2$, $S_n = S(\gamma)$ for some $\gamma \in (0,1) \cup (1,\infty)$, and $V_n = \nu_n \sigma^z$, where $(\nu_n)$ are i.i.d.\ random variables with non-trivial compactly supported distribution $\rho$.

Then, for all $E\ne0$, the F\"urstenberg group $G_{\mu_E}$ is Zariski-dense in $\Sp_2(\R)$. In particular, $\gamma_1(E)>\gamma_2(E)>0$ for all $E\ne0$.
\end{thm}

Before we prove Theorem \ref{Zariskidenseneq0}, let us lay some foundations. In the special case considered here the transfer matrices (\ref{tmprop}) can be factored as
\begin{equation} 
A_n^E = 
\left(
\begin{array}{cccc}
1 & 0 & 0 & 0 \\
0 & 1 & 0 & 0 \\
\nu_n & 0 & 1 & 0 \\
0 & -\nu_n & 0 & 1 \\
\end{array}
\right)
\left(
\begin{array}{cccc}
0 & 0 & \frac{1}{1-\gamma^2} & \frac{\gamma}{1-\gamma^2} \\
0 & 0 & -\frac{\gamma}{1-\gamma^2} & -\frac{1}{1-\gamma^2} \\
-1 & \gamma & -\frac{E}{1-\gamma^2} & -\frac{\gamma E}{1-\gamma^2} \\
-\gamma & 1 & \frac{\gamma E}{1-\gamma^2} & \frac{E}{1-\gamma^2} \\
\end{array}
\right),
\end{equation}
separating the random and energy parameters.

For any $2\times2$ matrix $Q$, if we define
\begin{equation} 
M(Q):=\left(
\begin{array}{cc}
I & 0 \\
Q & I \\
\end{array}
\right),
\end{equation}
then we may write the above factorization as
\begin{equation} 
A_n^E= M(\nu_n \sigma^z) A_0(E),
\end{equation}
where $A_0(E)$ is defined as the second factor. Thus the F\"urstenberg group is 
\begin{equation} 
G_{\mu_E} := \overline{\ip{\supp \mu_E}}= \overline{\langle M(\nu_0 \sigma^z)A_0(E):\nu_0\in\supp \rho\rangle}.
\end{equation}
For our proof, we will also need the following standard fact from Lie theory:

\begin{lemma}\label{Lieconj}
If $G$ is a matrix Lie group, and $\frak{g}$ is its Lie algebra, then we have
\begin{equation} 
AXA^{-1}\in\frak{g}
\end{equation}
whenever $X\in\frak{g}$ and $A\in G$.
\end{lemma}

\begin{proof}[Proof of Theorem \ref{Zariskidenseneq0}]
As in Subsection \ref{subsec:checking}, let $\frak{G}_2(E)$ denote the Lie algebra of $\Cl_Z(G_{\mu_E})$, and $\frak{sp}_2(\R)$ the Lie algebra of $\Sp_2(\R)$. As discussed there, it suffices to show that $\frak{G}_2(E)=\frak{sp}_2(\R)$ for all $E\ne0$. Let
\begin{equation}\label{UandblockU}
U=\frac{1}{\sqrt{2}}\left(
\begin{array}{cc}
1 & 1 \\
1 & -1 \\
\end{array}
\right)\quad\quad\text{and}\quad\quad
\UU:=\left(
\begin{array}{cc}
U & 0 \\
0 & U \\
\end{array}
\right).
\end{equation}
Because $U$ is symmetric and self-inverse, one easily checks that $\UU$ is symplectic. As a consequence of Lemma \ref{Lieconj}, one can show that
\begin{equation} 
\frak{G}_2(E)=\frak{sp}_2(\R) \quad \mbox{if and only if} \quad \tilde{\frak{G}}_2(E):=\UU \frak{G}_2(E) \UU=\frak{sp}_2(\R).
\end{equation}
It thus suffices to construct 10 linearly independent elements belonging to $\tilde{\frak{G}}_2(E)\subset\frak{sp}_2(\R)$. We will need the following tool that allows us to move back and forth between the Lie group and Lie algebra, see e.g.\ \cite{BS2006} for a proof.

\begin{lemma}\label{MQlemma}
For $Q\in\R^{2\times 2}$, we have
\begin{equation} 
M(Q)\in \Cl_Z(G_{\mu_E}) \quad \mbox{if and only if} \quad
\left(
\begin{array}{cc}
0 & 0 \\
Q & 0 \\
\end{array}
\right)
\in \frak{G}_2(E).
\end{equation}
\end{lemma}

We know that $M(\nu_0 \sigma^z)A_0(E)\in\Cl_Z(G_{\mu_E})$ whenever $\nu_0\in\supp \rho$. Our first objective will be to show that $A_0(E)\in\Cl_Z(G_{\mu_E})$.

As $\rho$ is non-trivial, let us take $a,b\in\supp \rho$, $a\ne b$. Then $M(a\sigma^z)A_0(E)\in G_{\mu_E}$ and $M(b\sigma^z)A_0(E)\in G_{\mu_E}$. This implies that
\begin{eqnarray}
M((a-b)\sigma^z) &=& M(a\sigma^z)A_0(E)[M(b\sigma^z)A_0(E)]^{-1} \label{tmcancel}\\
&\in& G_{\mu_E}\subset\Cl_Z(G_{\mu_E}). \nonumber
\end{eqnarray}
Lemma \ref{MQlemma} then implies
\begin{equation} 
\left(
\begin{array}{cc}
0 & 0 \\
(a-b)\sigma^z & 0 \\
\end{array}
\right)\in \frak{G}_2(E),
\end{equation}
and being a Lie algebra, all scalar multiples also lie in $\frak{G}_2(E)$. In particular,
\begin{equation}\label{diagcc}
\left(
\begin{array}{cc}
0 & 0 \\
c \sigma^z & 0 \\
\end{array}
\right)\in \frak{G}_2(E) \quad \mbox{for all $c\in\R$}.
\end{equation}
Again by Lemma \ref{MQlemma}, we have $M(a\sigma^z)\in\Cl_Z(G_{\mu_E})$ and thus
\begin{equation}\label{A0Eingroup}
A_0(E)=M(a\sigma^z)^{-1}[M(a\sigma^z)A_0(E)]\in\Cl_Z(G_{\mu_E}).
\end{equation}

We now construct the 10 linearly independent elements of $\tilde{\frak{G}}_2(E)$, which will complete the proof. In the following, any matrix with a superscript ``temp" will be replaced by a simpler matrix soon thereafter. We start with
\begin{equation} 
A_1 :=
\UU
\left(
\begin{array}{cc}
0 & 0 \\
\sigma^z & 0 \\
\end{array}
\right)
\UU
=
\left(
\begin{array}{cccc}
0 & 0 & 0 & 0 \\
0 & 0 & 0 & 0 \\
0 & 1 & 0 & 0 \\
1 & 0 & 0 & 0 \\
\end{array}
\right) \in \tilde{\frak{G}}_2(E).
\end{equation}
Conjugating via Lemma~\ref{Lieconj} yields
\begin{equation} 
A_2 := -(1-\gamma^2)
\UU
A_0(E)^{-1}
\left(
\begin{array}{cc}
0 & 0 \\
\sigma^z & 0 \\
\end{array}
\right)
A_0(E)
\UU
=
\left(
\begin{array}{cccc}
0 & 0 & 0 & 1 \\
0 & 0 & 1 & 0 \\
0 & 0 & 0 & 0 \\
0 & 0 & 0 & 0 \\
\end{array}
\right) \in \tilde{\frak{G}}_2(E)
\end{equation}
and
\begin{equation} 
A_3^{temp} := -(1-\gamma^2) \UU A_0(E)
\left(
\begin{array}{cc}
0 & 0 \\
\sigma^z & 0 \\
\end{array}
\right)
A_0(E)^{-1} \UU
=
\left(
\begin{array}{cccc}
0 & E & 0 & 1 \\
E & 0 & 1 & 0 \\
0 & -E^2 & 0 & -E \\
-E^2 & 0 & -E & 0 \\
\end{array}
\right)\in\tilde{\frak{G}}_2(E)
\end{equation}
By taking linear combinations of $A_1, A_2, A_3^{temp}$ and using that $E\ne0$, it is clear we can produce
\begin{equation} 
A_3 := 
\left(
\begin{array}{cccc}
0 & 1 & 0 & 0 \\
1 & 0 & 0 & 0 \\
0 & 0 & 0 & -1 \\
0 & 0 & -1 & 0 \\
\end{array}
\right) \in \tilde{\frak{G}}_2(E).
\end{equation}
Furthermore, if we define $b_1:=E+\gamma+1$, $b_2:=E-\gamma+1$, $b_3:=E+\gamma-1$, and $b_4:=E-\gamma-1$, then
\begin{eqnarray}
A_4^{temp} &:=& -(1-\gamma^2)^2 \UU A_0(E)^2
\left(
\begin{array}{cc}
0 & 0 \\
\sigma^z & 0 \\
\end{array}
\right)
A_0(E)^{-2} \UU \\
&=& \left(
\begin{array}{cccc}
0 & Eb_2b_3 & 0 & E^2 \\
Eb_1b_4 & 0 & E^2 & 0 \\
0 & -b_1b_2b_3b_4 & 0 & -Eb_1b_4 \\
-b_1b_2b_3b_4 & 0 & -Eb_2b_3 & 0 \\
\end{array}
\right) \: \in \: \tilde{\frak{G}}_2(E) \nonumber
\end{eqnarray}
Since $E\ne 0$ and $\gamma\ne 0$, one can see that $Eb_1b_4\ne Eb_2b_3$. After separating into the five cases $b_1=0$, $b_2=0$, $b_3=0$, $b_4=0$, and all $b_j\ne 0$, it is not difficult to find a linear combination of $A_1$, $A_2$, $A_3$, and $A_4^{temp}$ that yields
\begin{equation} 
A_4 := 
\left(
\begin{array}{cccc}
0 & 1 & 0 & 0 \\
0 & 0 & 0 & 0 \\
0 & 0 & 0 & 0 \\
0 & 0 & -1 & 0 \\
\end{array}
\right) \in \tilde{\frak{G}}_2(E).
\end{equation}

The matrices $A_1,A_2,A_3,A_4$ are linearly independent and span the following 4-dimensional subspace of $\frak{sp}_2(\R)$:
\begin{equation} 
\left\{
\left(
\begin{array}{cccc}
0 & a & 0 & d \\
b & 0 & d & 0 \\
0 & c & 0 & -b \\
c & 0 & -a & 0 \\
\end{array}
\right): a,b,c,d\in\R
\right\}.
\end{equation}

It turns out, and is easily checked by calculation, that the complementary 6-dimensional subspace
\begin{equation} \label{compspace}
\left\{
\left(
\begin{array}{cccc}
a & 0 & d_1 & 0 \\
0 & b & 0 & d_2 \\
c_1 & 0 & -a & 0 \\
0 & c_2 & 0 & -b \\
\end{array}
\right): a,b,c_1,c_2,d_1,d_2\in\R
\right\}
\end{equation}
of $\frak{sp}_2(\R)$ is spanned by the six matrix commutators $[A_{j_1},A_{j_2}]$, $1\le j_1<j_2\le 4$, of the first four matrices. Here, the commutator $[A,B]:=AB-BA$ is the Lie bracket of the Lie algebra $\tilde{\frak{G}}_2(E)$. This completes the construction of $10$ linearly independent elements of $\tilde{\frak{G}}_2(E)$ and thus the proof of Theorem~\ref{Zariskidenseneq0}.
\end{proof}

\subsection{Critical energy $E=0$} \label{subsec:crit}

We have just shown that for all $E \ne 0$ the F\"urstenberg group $G_{\mu_E}$ is Zariski-dense in $\Sp_2(\R)$ and thus, by the Gol'dsheid-Margulis criterion, $p$-contracting and $L_p$-strongly irreducible for $p=1$ and $p=2$. It is easy to see that at $E=0$ the F\"urstenberg group is not strongly irreducible in $\R^4$ (which is the same as saying that it is not $L_1$-strongly irreducible). Thus (\ref{DefCI}) does not hold for any interval $I$ containing $0$.

However, we can still show by a direct argument, see part (ii) of the following result, that the leading Lyapunov exponent $\gamma_1(0)$ is strictly positive and distinct from $\gamma_2(0)$. After proving this we will discuss that typically $\gamma_2(0)$ is also positive, while there are exceptional cases where it may vanish.

\begin{thm} \label{thm:zeroenergy}
Under the conditions of Theorem~\ref{Zariskidenseneq0},
\begin{itemize}
\item[(i)] $G_{\mu_0}$ is not strongly irreducible in $\R^4$, and
\item[(ii)] $\gamma_1(0)>\gamma_2(0)\ge0$.
\end{itemize}
\end{thm}
\begin{proof}

(i) Since $G_{\mu_0}$ is the smallest closed subgroup of $\Sp_2(\R)$ containing the transfer matrices
\begin{equation}
A_n^0=\left(
\begin{array}{cccc}
0 & 0 & \frac{1}{1-\gamma^2} & \frac{\gamma}{1-\gamma^2} \\
0 & 0 & -\frac{\gamma}{1-\gamma^2} & -\frac{1}{1-\gamma^2} \\
-1 & \gamma & \frac{\nu_n}{1-\gamma^2} & \frac{\gamma \nu_n}{1-\gamma^2} \\
-\gamma & 1 & \frac{\gamma\nu_n}{1-\gamma^2} & \frac{\nu_n}{1-\gamma^2} \\
\end{array}
\right), \quad n\in\N,
\end{equation}
it follows that $\tilde G_{\mu_0}:=\UU G_{\mu_0}\UU$ (with $\UU$ from (\ref{UandblockU})) is the smallest closed subgroup of $\Sp_2(\R)$ containing the transformed transfer matrices
\begin{equation}
B_n^0:=\UU A_n^0\UU=\left(
\begin{array}{cccc}
0 & 0 & 0 & \frac{1}{1+\gamma} \\
0 & 0 & \frac{1}{1-\gamma} & 0 \\
0 & -1-\gamma & \frac{\nu_n}{1-\gamma} & 0 \\
-1+\gamma & 0 & 0 & \frac{\nu_n}{1+\gamma} \\
\end{array}
\right), \quad n\in\N.
\end{equation}
All of these matrices and, as a consequence, all elements of $\tilde G_{\mu_0}$ are of the form
\begin{equation} 
\left(
\begin{array}{cccc}
a_{11} & 0 & 0 & a_{14} \\
0 & a_{22} & a_{23} & 0 \\
0 & a_{32} & a_{33} & 0 \\
a_{41} & 0 & 0 & a_{44} \\
\end{array}
\right).
\end{equation}
Now it is obvious that $\tilde{G}_{\mu_0}$ (and thus $G_{\mu_0}$) has non-trivial invariant subspaces in $\R^4$ and therefore is not strongly irreducible.

(ii) The upshot of the above argument is that the transfer matrices $A_n^0$ are similar to direct sums of $2\times 2$ transfer matrices, allowing to analyze them more directly with the classical F\"urstenberg theorem. More precisely, if
\begin{equation}
 P=\left(
\begin{array}{cccc}
1 & 0 & 0 & 0 \\
0 & 0 & 1 & 0 \\
0 & 0 & 0 & 1 \\
0 & 1 & 0 & 0 \\
\end{array}
\right),
\end{equation}
then
\begin{equation}
C_n = P^{-1} B_n^0 P = P^{-1} \UU A_n^0 \UU P = \left(
\begin{array}{cccc}
0 & \frac{1}{1+\gamma} & 0 & 0 \\
-1+\gamma & \frac{\nu_n}{1+\gamma} & 0 & 0 \\
0 & 0 & 0 & \frac{1}{1-\gamma} \\
0 & 0 & -1-\gamma & \frac{\nu_n}{1-\gamma} \\
\end{array}
\right)=\left(
\begin{array}{cc}
D_n & 0 \\
0 & F_n \\
\end{array}
\right),
\end{equation}
where
\begin{equation} 
D_n:=\left(
\begin{array}{cc}
0 & \frac{1}{1+\gamma} \\
-1+\gamma & \frac{\nu_n}{1+\gamma} \\
\end{array}
\right)\quad\text{and}\quad
F_n:=\left(
\begin{array}{cc}
0 & \frac1{1-\gamma} \\
-1-\gamma & \frac{\nu_n}{1-\gamma} \\
\end{array}
\right).
\end{equation}

Due to symplecticity, the Lyapunov exponents of $\set{A_n^0}_{n\in\N}$ are of the form $\gamma_1(0) \ge \gamma_2(0) \ge 0 \ge -\gamma_2(0) \ge -\gamma_1(0)$. 

For the proof of $\gamma_1(0) > \gamma_2(0)$ we will first consider the case $\gamma \in (0,1)$. Let $\gamma_1^D$ and $\gamma_2^D$ be the Lyapunov exponents of $\set{D_n}_{n\in\N}$ and $\gamma_1^F$ and $\gamma_2^F$ the Lyapunov exponents of $\set{F_n}_{n\in\N}$ (whose existence follows by the argument below). Then we have
\begin{equation}\label{Lyapseteq1}
\set{\gamma_1(0),\gamma_2(0),-\gamma_2(0),-\gamma_1(0)}=\set{\gamma_1^D,\gamma_2^D,\gamma_1^F,\gamma_2^F}.
\end{equation}

Let us define new matrices
\begin{eqnarray}
\tilde D_n &:=& \sqrt{\frac{1+\gamma}{1-\gamma}}D_n=\left(
\begin{array}{cc}
0 & \frac{1}{\sqrt{1-\gamma^2}} \\
-\sqrt{1-\gamma^2} & \frac{\nu_n}{\sqrt{1-\gamma^2}} \\
\end{array}
\right) \\
\tilde F_n &:=& \sqrt{\frac{1-\gamma}{1+\gamma}}F_n=\left(
\begin{array}{cc}
0 & \frac{1}{\sqrt{1-\gamma^2}} \\
-\sqrt{1-\gamma^2} & \frac{\nu_n}{\sqrt{1-\gamma^2}} \\
\end{array}
\right)=\tilde D_n.
\end{eqnarray}
Since $\tilde D_n$ is a constant multiple of $D_n$, the logarithm in the definition of the Lyapunov exponent turns this constant multiple into a shift, and (\ref{Lyapseteq1}) then becomes
\begin{equation}\label{Lyapseteq2}
\set{\gamma_1(0),\gamma_2(0),-\gamma_2(0),-\gamma_1(0)}
=\set{\pm\gamma^{\tilde D}\pm\frac12\log\frac{1+\gamma}{1-\gamma}}.
\end{equation}

We note that $\det\tilde D_n= \det\tilde F_n=1$ and that the matrices $\tilde{D}_n$ are similar to
\begin{equation}\label{Dnprime}
D_n' = \left(
\begin{array}{cc}
1 & 0 \\
0 & \frac{1}{\sqrt{1-\gamma^2}} \\
\end{array}
\right)\tilde D_n\left(
\begin{array}{cc}
1 & 0 \\
0 & \sqrt{1-\gamma^2} \\
\end{array}
\right)=\left(
\begin{array}{cc}
0 & 1 \\
-1 & \frac{\nu_n}{\sqrt{1-\gamma^2}} \\
\end{array}
\right),
\end{equation}
which are the transfer matrices of the standard Anderson model at $E=0$, with disorder scaled by the constant $1/\sqrt{1-\gamma^2}$. It is well known and follows from F\"urstenberg's theorem that $\gamma^{\tilde D}>0$.

Noting that $(1+\gamma)/(1-\gamma) >1$ and thus $\log \frac{1+\gamma}{1-\gamma}>0$, we conclude that 
\begin{equation} \label{FirstTwo}
\gamma_{1}(0) = \gamma^{\tilde D} + \frac12\log\frac{1+\gamma}{1-\gamma}, \quad \gamma_2(0) = \left|\gamma^{\tilde D}-\frac12\log\frac{1+\gamma}{1-\gamma}\right|,
\end{equation}
and, in particular, $\gamma_1(0)> \gamma_2(0)$.

In the case $\gamma>1$ the matrices $D_n$ and $F_n$ have negative determinant, which leads us to modifying the above argument by considering products of neighboring pairs of transfer matrices. Thus let
\begin{eqnarray}
G_n & := & D_{2n} D_{2n-1} = \frac{\gamma-1}{\gamma+1} \tilde{G}_n, \\
H_n & := & F_{2n} F_{2n-1} = \frac{\gamma+1}{\gamma-1} \tilde{G}_n,
\end{eqnarray}
where
\begin{equation}\label{twosteptm}
\tilde{G}_n = \left( \begin{array}{cc} 1 & \frac{\nu_{2n-1}}{\gamma^2-1} \\ \nu_{2n} & 1 + \frac{\nu_{2n-1} \nu_{2n}}{\gamma^2-1} \end{array} \right).
\end{equation}
The matrices $\tilde{G}_n$ are i.i.d.\ with $\det \tilde{G}_n =1$.

Observe that 
\begin{equation}
C_{2n} C_{2n-1} = \left( \begin{array}{cc} G_n & 0 \\ 0 & H_n \end{array} \right),
\end{equation}
and thus
\begin{equation}\label{Lyapseteq3}
\set{\gamma_1(0),\gamma_2(0),-\gamma_2(0),-\gamma_1(0)}
=\frac{1}{2} \set{\pm\gamma^{\tilde G}\pm\frac12\log\frac{\gamma+1}{\gamma-1}}.
\end{equation}
As above, we argue that
\begin{equation}
\gamma_1(0) = \frac{1}{2} \left( \gamma^{\tilde{G}} + \frac12\log \frac{\gamma+1}{\gamma-1} \right), \quad \gamma_2(0) = \frac{1}{2} \left| \gamma^{\tilde{G}} - \frac12\log \frac{\gamma+1}{\gamma-1} \right|
\end{equation}
To conclude $\gamma_1(0)> \gamma_2(0)$, it remains to show that $\gamma^{\tilde G}>0$, which we accomplish by verifying the assumptions of F\"urstenberg's theorem.

Associated to (\ref{twosteptm}), putting $c:=\gamma^2-1$, we define the F\"urstenberg group
\begin{equation}
G=\overline{\ip{G(x,y):=\left(\begin{array}{cc} 1 & \frac{x}{c} \\ y & 1+\frac{xy}{c}\end{array}\right):x,y\in\supp\rho}}\subset\SL_2(\R).
\end{equation}
Pick $\{a, b\} \subset \supp \rho$ with $a \not= b$. Then a short calculation shows that
\begin{equation}
G(a,a)^{-1} G(b,a) = \left( \begin{array}{cc} 1 & \frac{b-a}{c} \\ 0 & 1 \end{array} \right)\in G.
\end{equation}

Taking successive powers of this matrix causes the upper-right entry to grow, showing non-compactness of $G$. One also proves strong irreducibility in much the same manner as it is done for the Anderson model (see, e.g., Proposition~IV.4.25 in \cite{CL1990}):

Let $v=(\alpha,\beta)^t\in\R^2\setminus\set{0}$. If $\beta\ne0$, one checks that the three vectors
\begin{equation}
v,\quad \left( \begin{array}{cc} 1 & \frac{b-a}{c} \\ 0 & 1 \end{array} \right)v,\quad \left( \begin{array}{cc} 1 & \frac{b-a}{c} \\ 0 & 1 \end{array} \right)^2v
\end{equation}
are pairwise non-collinear. For $(\alpha,\beta) = (1,0)$ one takes $y\in\set{a,b}\setminus\set{0}$, defines $w=G(a,y)v=(1,y)^t$, and finds that the vectors
\begin{equation}
w,\quad \left( \begin{array}{cc} 1 & \frac{b-a}{c} \\ 0 & 1 \end{array} \right)w,\quad \left( \begin{array}{cc} 1 & \frac{b-a}{c} \\ 0 & 1 \end{array} \right)^2 w
\end{equation}
are pairwise non-collinear. Thus, F\"urstenberg's theorem gives $\gamma^{\tilde G}>0$, as desired. This completes the proof.
\end{proof}

We now know that $\gamma_1(0)>\gamma_2(0)$, but deciding if $\gamma_2(0)>0$ is more subtle and depends on the specific choice of parameters of the model. We will only discuss this for the case $\gamma \in (0,1)$, where we can cite known facts for the Anderson model, but similar reasoning should apply for $\gamma>1$.

If $\gamma \in (0,1)$, then we conclude from (\ref{FirstTwo}) that $\gamma_2(0)>0$ if and only if
\begin{equation}
\gamma^{\tilde{D}} \not= \frac{1}{2} \log \frac{1+\gamma}{1-\gamma}.
\end{equation}
As we mentioned after (\ref{Dnprime}), we know that $\gamma^{\tilde{D}}>0$, but it depends on the specifics of the distribution $\rho$ of $\nu_n$ if $\gamma^{\tilde{D}}$ may happen to coincide with $\frac{1}{2} \log \frac{1+\gamma}{1-\gamma}$.

To illustrate this further, introduce an extra disorder parameter $\alpha$ and let $\pm \Gamma_0(\gamma,\alpha)$ be the Lyapunov exponents of the i.i.d.\ matrices
\begin{equation} 
D_n'(\alpha)=\left(
\begin{array}{cc}
0 & 1 \\
-1 & \frac{\alpha\nu_n}{\sqrt{1-\gamma^2}} \\
\end{array}
\right).
\end{equation}
We claim that
\begin{eqnarray}
& \alpha \mapsto  \Gamma_0(\gamma,\alpha) \quad\text{is continuous},& \\
& \lim_{\alpha\to0}\Gamma_0(\gamma,\alpha) = 0, & \\
& \lim_{\alpha\to\infty}\Gamma_0(\gamma,\alpha) = \infty. &
\end{eqnarray}
The first result is proven in the same way as showing continuity of the Lyapunov exponent as a function of energy, compare, e.g., Corollary~V.4.8 of \cite{CL1990} and its proof. The second and third results follow from asymptotic relations of the Lyapunov exponent for the Anderson model. For this, we refer to Sections V.11 and VI.14 in \cite{PF1992}.

Since $\frac12\log\frac{1+\gamma}{1-\gamma}>0$, the Intermediate Value Theorem yields for any given non-trivial distribution $\rho$ of the $\nu_n$ a choice of $\alpha$ (and hence a re-scaled distribution $\rho_{\alpha} = \rho(\cdot/\alpha)$) such that $\gamma_2(0)=0$. We also get $\gamma_2(0)>0$ if $\alpha$ is sufficiently large or sufficiently small. If $\Gamma_0(\gamma,\alpha)$ were strictly monotone in $\alpha$ (which we haven't checked), then $\gamma_2(0)$ would vanish for a unique critical value of $\alpha$ and, in this sense, $\gamma_2(0)>0$ would be the generic case.

\subsection{Other types of randomness} \label{sec:otherrand}

In Theorems~\ref{thm:main} and \ref{Zariskidenseneq0} we have focused on the special case of the model (\ref{blockmatrixhatM}) in which the magnetic field strength $\nu_n$ is random. But, of course, the more general Theorem~\ref{BJMdynloc} together with the Gol'dsheid-Margulis criterion can be applied to other types of randomness, as long as Zariski-denseness of the F\"urstenberg groups can be verified for suitable energy intervals. Without stating detailed results, we discuss here what we found for the cases where in the model (\ref{blockmatrixhatM}) either the coupling constants $\mu_n$ or the anisotropy parameters $\gamma_n$ are chosen to be random. This was done, in part, by using numerical help. Note that cases with randomness in two or all three of the parameter sequences are easier, as this gives larger F\"urstenberg groups which are more likely to be Zariski-dense.

First, note that the zero-energy singularity will persist for all these cases, as the reducibility of the transfer matrices observed in the proof of Theorem~\ref{thm:zeroenergy} holds in general for the transfer matrices $A_k^0$ from (\ref{tmprop}) with $V_k$ and $S_k$ as in (\ref{blockmatrixhatM}).

In the case of random i.i.d.\ couplings $\mu_n$ taking at least two different values, keeping $\nu_n=\nu$ and $\gamma_n=\gamma\in(0,1)$ constant, we found, using a construction similar to the proof of Theorem~\ref{Zariskidenseneq0}, that the F\"urstenberg group is Zariski-dense for all $E\notin\set{0,\pm\nu}$. But we have not checked if the new critical energies $\pm \nu$ lead to interesting phenomena and if they might disappear if the support of the single-site distribution of the $\mu_n$ has more than two points.

In the case of random i.i.d.\ anisotropies $\gamma_n$ with support contained in $(0,1)$, keeping $\nu_n=\nu$ and $\mu_n=1$ constant, the analysis is more complex. The standard canceling trick (\ref{tmcancel}) no longer results in a matrix of the form $M(Q)$ but rather a diagonal matrix (with positive diagonal entries), so we may no longer use Lemma \ref{MQlemma} to cleverly switch back and forth between the Lie algebra and Lie group. Instead, under the assumption that the single-site distribution takes two different values $a$ and $b$, ten elements in the Lie algebra were constructed by appropriately conjugating the logarithm of this diagonal matrix, which is known to lie in the Lie algebra. Numerically, for any non-zero $E$, a transcendental equation (arising as a 10$\times$10-determinant for the ten constructed elements of the Lie algebra) has no more than two roots in $b$ for any fixed $a$. This would indicate that Zariski-denseness holds for all $E\ne 0$ if the single-site distribution of the $\gamma_n$ is supported on at least four points, as in this case for any value of the first point at least one of the three other values would yield non-zero determinant.

\section{Dynamical localization for the XY chain} \label{sec:discussion}

Our interest in proving dynamical localization for the random block operators $\hat{M}_n$ defined in (\ref{blockmatrixM}) is largely motivated by the fact that, as proven in \cite{HSS2012}, this implies a many-body dynamical localization property for the disordered XY spin chain. In this concluding section we discuss this connection between one-body and many-body localization in more detail.

The identity (\ref{CMC}), derived via the Jordan-Wigner transform, can be read as relating the many-body Hamiltonian $H_n$ to the effective one-particle Hamiltonian $\hat{M}_n$. In connection with the CAR (\ref{CAR}) this leads to a relation between the Heisenberg dynamics of $H_n$ and the Schr\"odinger dynamics of $\hat{M}_n$,
\begin{equation}\label{tauM}
\tau_t^n(c_j)=\sum_{k=1}^n \hat M_{n,j,k}(2t)c_k+\sum_{k=1}^n \hat M_{n,j,n+k}(2t)c_k^*,\quad\quad j=1,...,n,
\end{equation}
see, e.g., \cite{HSS2012} for a proof. Here $\tau_t^n(a):=e^{itH_n}ae^{-itH_n}$ is the Heisenberg dynamics for an operator $a$ on $\bigotimes^n \C^2$, and $\hat M_{n,j,k}(t):=(e^{-it\hat M_n})_{j,k}$ is the $(j,k)$-th matrix element of the time evolution of $\hat M_n$. 

The identity (\ref{tauM}) is the key fact which allows to turn results on dynamical localization for $\hat{M}_n$ into dynamical localization properties of $H_n$. More precisely, we have the following result. Here the notation $\mathcal{A}_N$, for $N\subset [1,n]$, represents the class of tensor product operators on $\bigotimes^n \C^2$ which act trivially on sites outside $N$. For more background and precise definitions, including discussion of the interpretation of bounds of the form (\ref{LRbound}) below as {\it zero-velocity Lieb-Robinson bounds} for many-body systems, we refer to \cite{HSS2012}.

\begin{thm}\label{LRthm}
Suppose there exist $\zeta\in(0,1)$ and constants $C>0$, $\eta>0$ such that for all $n\in\N$ and $j,k\in[1,n]$,
\begin{equation}\label{subDL}
\E\left(\sup_{t\in\R} \|P_j e^{-itM_n} P_k^*\| \right)\le Ce^{-\eta\abs{j-k}^\zeta}.
\end{equation}
Then for every $\eps\in(0,\eta)$, there exists a constant $C'=C'(\eta,\eps,\zeta)>0$ such that 
\begin{equation}\label{LRbound}
\E\left(\sup_{t\in\R}\|[\tau_t^n(A),B]\|\right)\le C^\prime\norm{A}\norm{B}e^{-(\eta-\eps)(k-j)^\zeta}
\end{equation}
for all $1\le j<k$, $n\ge k$, $A\in\mathcal{A}_j$ and $B\in\mathcal{A}_{[k,n]}$. Furthermore, if (\ref{subDL}) holds with $\zeta=1$, then (\ref{LRbound}) holds with $\eps=0$.
\end{thm}
\begin{proof}
The last statement regarding $\zeta=1$ and $\eps=0$ is exactly Theorem 3.2 of \cite{HSS2012}. The proof of the result for $\zeta<1$ requires a slight modification, and is presented in Chapter 2 of \cite{Chap2013}. The proof with $\zeta=1$ uses that the tail $\sum_{j=k}^\infty e^{-cj}$ of a geometric series of exponential terms is proportional to the first term $e^{-ck}$. For $\zeta\in(0,1)$, one has to properly estimate the tail $\sum_{j=k}^\infty e^{-cj^\zeta}$ in terms of the first term $e^{-ck^\zeta}$, which is done via an integral comparison argument, leading to a slight loss of the decay rate in form of an $\eps>0$. Also note that \cite{HSS2012} and \cite{Chap2013} state their results with $\sup_t |\hat{M}_{n,j,k}(t)| + \sup_t |\hat{M}_{n,j,n+k}(t)|$ instead of $\sup_t \|P_j e^{-itM_n} P_k^*\|$ in (\ref{subDL}), but these two terms are equivalent. 
\end{proof}

There are two regimes in which (\ref{subDL}) has been verified previously, in both cases with $\zeta=1$:

(i) For the {\it isotropic} XY chain, i.e.\ for $\gamma=0$, the block operator $M_n$ reduces to the one-dimensional Anderson model, for which (\ref{subDL}) with $\zeta=1$ is well known (under the assumption (\ref{eq:coeffcond2}) on the single-site distribution). The corresponding Lieb-Robinson bound
\begin{equation} \label{eq:LRzeta1}
\E\left(\sup_{t\in\R}\|[\tau_t^n(A),B]\|\right)\le C^\prime\norm{A}\norm{B}e^{-\eta(k-j)},
\end{equation}
with $A$ and $B$ as above, was found in \cite{HSS2012}.

(ii) In the general anisotropic case $\gamma \not= 0$ the large disorder regime (replace $\nu_n$ with $\lambda \nu_n$ and $\lambda>0$ sufficiently large) is covered by a special case of the results in \cite{ESS2012}, mentioned earlier in Section~\ref{sec:xychain}. There a dynamical localization bound of the form (\ref{subDL}) with $\zeta=1$ is proven for a much larger class of random block operators, also covering multi-dimensional block operators, but requiring large disorder as well as sufficiently smooth distribution of the random parameters.  Again, via Theorem~\ref{LRthm}, this implies dynamical localization in the form (\ref{eq:LRzeta1}).

The following new result on zero-velocity Lieb-Robinson bounds for the disordered XY chain can be found by combining Theorem~\ref{thm:main} with Theorem~\ref{LRthm}. It allows for singular distributions of the $\nu_j$, but still requires that they are large in suitable sense (but not in the sense of large disorder as in \cite{ESS2012}).

\begin{thm} \label{thm:dynlocanisoXY}
Assume that the parameters in the XY chain Hamiltonian (\ref{xychain1}) satisfy (\ref{eq:coeffcond1}) and (\ref{eq:coeffcond2}) and, in addition, that supp$\,\rho$, the compact support of the distribution of the $\nu_j$, is contained either in $(2,\infty)$ or in $(-\infty,-2)$. 

Then for every $\zeta \in (0,1)$ there exists $C=C(\zeta)<\infty$ such that
\begin{equation} \label{eq:newLRbound}
\E \left( \sup_{t\in \R} \| [ \tau_t^n(A), B]\| \right) \le C \|A\| \|B\| e^{-(k-j)^{\zeta}}
\end{equation}
for all $1\le j \le k \le n$, $A\in\mathcal{A}_j$ and $B\in\mathcal{A}_{[k,n]}$. 
\end{thm}

\begin{proof}
The crucial fact is that under the additional assumption supp$\,\rho \subset (2,\infty)$ or supp$\,\rho \subset (-\infty,-2)$ the operators $M_n$ have a spectral gap at $E=0$, allowing to remove the spectral projection $\chi_J(M_n)$ in (\ref{DL2}). More precisely, choose
\begin{equation}
\lambda = \left\{ \begin{array}{ll} a-2, & \mbox{if $a= \min(\rm{supp}\,\rho) > 2$}, \\ -2-b, & \mbox{if $b= \max(\rm{supp}\,\rho) < -2$}. \end{array} \right.
\end{equation}
Then it is easily seen that the diagonal block $A_n$ of $\hat{M}_n$ in (\ref{blockmatrixM}) almost surely satisfies $\sigma(A_n) \subset [\lambda,\infty)$ or $\sigma(A_n) \subset (-\infty,\lambda]$, respectively. By (a finite volume version) of Proposition~\ref{specsymlemma} this implies $\sigma(M_n) \cap (-\lambda, \lambda) = \emptyset$ for all $n$ and almost all $(\nu_j)$. Thus there are compact intervals $J_1 \subset (0,\infty)$ and $J_2 \subset (-\infty,0)$ (in fact, $J_2 = - J_1$), such that $\sigma(M_n) \subset J_1 \cup J_2$. Thus, by Theorem~\ref{thm:main}, for all $\zeta \in (0,1)$ there are $C<\infty$ and $\eta>0$ such that 
\begin{eqnarray}
\lefteqn{\E \left( \sup_t \|P_j e^{-itM_n} P_k^*\| \right)} \\
& \le & \E \left( \sup_t \|P_j e^{-itM_n} \chi_{J_1}(M_n) P_k^*\| \right) + \E \left( \sup_t \|P_j e^{-itM_n} \chi_{J_2}(M_n)P_k^*\| \right) \nonumber \\
& \le & C e^{-\eta|j-k|^{\zeta}} \nonumber
\end{eqnarray}
for all $1\le j \le k \le n$. Now Theorem~\ref{LRthm} implies (\ref{eq:newLRbound}), where $\eta-\varepsilon$ in (\ref{LRbound}) can be absorbed into the constant $C$ by slightly reducing $\zeta$.
\end{proof}

Removing $\chi_J(M_n)$ in (\ref{DL2}) in situations where $E=0$ does not lie in a spectral gap (and thus removing the extra assumption on supp$\,\rho$ in Theorem~\ref{thm:dynlocanisoXY}) is a more challenging problem.

Our discussion in Section~\ref{subsec:crit} is meant to be a first step towards understanding this, at least for generic single-site distributions $\rho$. We have shown there that, despite the lack of irreducibility of the F\"urstenberg group $G_{\mu_0}$, one can still show in many cases that $\gamma_1(0)>\gamma_2(0)>0$. However, this alone is not enough to show that $M_n$ is dynamically localized at {\it all} energies, including near zero. Irreducibility of the F\"urstenberg groups $G_{\mu_E}$ enters the proof of Theorem~\ref{thm:main} one more time, namely in the proof of H\"older continuity (\ref{gammaHolder}) of the Lyapunov exponents. The argument in \cite{CL1990} which we have referred to in this context relies strongly on the uniqueness of invariant measures associated with the transfer matrices in (\ref{eq:defLyap}), which in turn depends on irreducibility. 

It is thus a non-trivial and interesting question if H\"older continuity of the Lyapunov exponents can be shown near the critical energy $E=0$, at which the F\"urstenberg group becomes reducible. 

We plan to address questions of this type in \cite{IsingPaper}, where we will first consider the case of the Ising model, i.e.\ $\gamma=1$ in the XY chain. The reason that this was excluded in the present work is that in this case the matrices $S(\gamma)$ in (\ref{defSgamma}) are {\it not} invertible. It turns out, however, that in this case the block Jacobi matrices $M_n$ can be reduced, via reducing $S(1)$ to Jordan form, to standard Jacobi matrices (a fact well known in the physics literature since \cite{Pf1970}). For the latter one can show dynamical localization at non-zero energies with similar (in fact simpler) methods than were used here, see \cite{Chap2013}. But $E=0$ is once again a critical energy at which the transfer matrices become reducible (in fact, diagonal). In \cite{IsingPaper} we will provide an explicit argument that, under suitable assumptions on the single-site distribution $\rho$, the Lyapunov exponents are Lipschitz continuous in a neighborhood of zero, thus allowing to conclude a dynamical localization bound of the form (\ref{LRbound}) for the Ising model in random transversal field (and not requiring that $E=0$ lies in a spectral gap of the effective one-particle Hamiltonian).

%\appendix
%\section{}


\bigskip
\begin{thebibliography}{A}

\bibitem{Baskoetal} D.\ M.\ Basko, I.\ L.\ Aleiner and B.\ L.\ Altshuler, Metal-insulator transition in a weakly interacting many-electron system with localizad single-particle states, Annals of Physics {\bf 321} (2006), 1126--1205 

\bibitem{BL1985} P.\ Bougerol and J.\ Lacroix, Products of random matrices with applications to Schr\"odinger operators, Birkh\"auser, Boston, 1985

\bibitem{BS2006} H.\ Boumaza and G.\ Stolz, Positivity of Lyapunov exponents for Anderson-type models on two coupled strings, Elec.~J.~Diff.~Eq.~{\bf 2007}, no. 47, (2007), 1--18

\bibitem{Bou2009} H.\ Boumaza, Localization for a matrix-valued Anderson model, Math.~Phys.~Anal.~Geom.\ {\bf 12}, no.~3 (2009), 225--286

\bibitem{BM2013} H.\ Boumaza and L.\ Marin, Absence of absolutely continuous spectrum for random scattering zippers, Preprint, arXiv:1303.3116

\bibitem{BravyiKoenig} S.\ Bravyi and R.\ K\"onig, Disorder-assisted error correction in Majorana chains, Comm.~Math.~Phys.~{\bf 316} (2012), 641--692

\bibitem{CKM1987} R.\ Carmona, A.\ Klein and F.\ Martinelli, Anderson localization for Bernoulli and other singular potentials, Comm.~Math.~Phys.~{\bf 108} (1987), 41--66

\bibitem{CL1990} R. Carmona and J. Lacroix, Spectral theory of random Schr\"odinger operators, {\it Probability Theory and its Applications}, Birkh\"auser, Boston, 1990

\bibitem{Chap2013} J.\ Chapman, Spectral Properties of Random Block Operators, PhD Thesis, University of Alabama at Birmingham, 2013, electronically available at http://gradworks.umi.com/3561259.pdf

\bibitem{IsingPaper} J.\ Chapman and G.\ Stolz, Dynamical localization for the quantum Ising model in random field, in preparation

\bibitem{CS1983} W.\ Craig and B.\ Simon, Log H\"older continuity of the integrated density of states for stochastic Jacobi matrices, Comm.~Math.~Phys.~{\bf 90} (1983), 207--218

\bibitem{CS1983S} W.\ Craig and B.\ Simon, Subharmonicity of the Lyaponov index, Duke Math.~J.~{\bf 50}, no. 2, (1983), 551--560

\bibitem{ES2013} A.\ Elgart and D.\ Schmidt, Eigenvalue statistics for random block operators, Preprint, arXiv:1306.3459 

\bibitem{ESS2012} A.\ Elgart, M.\ Shamis, and S.\ Sodin, Localisation for non-monotone Schr\"odinger operators, Preprint, arXiv:1201.2211

\bibitem{GM2013} M.\ Gebert and P.\ M\"uller, Localization for random block operators, Oper.~Theory Adv.~Appl.~{\bf 232} (2013), 229--246

\bibitem{GeKle2001} F. Germinet and A. Klein, Bootstrap multiscale analysis and localization in random media, Comm.~Math.~Phys.~{\bf 222} (2001), 415--448

\bibitem{GeKlo2012} F.\ Germinet and F.\ Klopp, Spectral statistics for random Schr\"odinger operators in the localized regime, Preprint, arXiv:1011.1832

\bibitem{GoMa1989} I.\ Gol'dsheid and G.\ Margulis, Lyapunov indices of a product of random matrices, Russian Math.~Survey~{\bf 44:5} (1989), 11--71

\bibitem{HSS2012} E.\ Hamza, R.\ Sims and G.\ Stolz, Dynamical Localization in 
Disordered Quantum Spin Systems, Commun.~Math.~Phys.~{\bf 315} (2012), 215--239

\bibitem{Kir1989} W.\ Kirsch, Random Schr\"odinger operators. Schr\"odinger operators, Proc. Nord. Summer Sch. Math., Sandbjerg Slot, Sonderborg/Denmark 1988, Lect. Notes Phys.~{\bf 345} (1989), 264--370

\bibitem{Kir2007} W.\ Kirsch, An invitation to random Schr\"odinger operators, Panor.~Synth\'esis {\bf 25}, Random Schr\"odinger operators, 1--119, Soc.~Math.~France, Paris, 2008

\bibitem{KMM2010} W.\ Kirsch, B.\ Metzger and P.\ M\"uller, Random block operators, J.~Stat.~Phys.~{\bf 143}, no. 6, (2011), 1035--1054

\bibitem{Kitaev} A.\ Yu.\ Kitaev, Unpaired Majorana fermions in quantum wires, Phys.-Usp.\ {\bf 44} (2001), 131--136, see also arXiv:cond-mat/0010440

\bibitem{Klein2008} A.\ Klein, Multiscale analysis and localization of random operators. Random Schr\"odinger operators, 121-159, Panor.~Synth\'esis, {\bf 25}, Soc.~Math.~France, Paris, 2008

\bibitem{KLS1990} A.\ Klein, J.\ Lacroix and A.\ Speis, Localization for the Anderson model on a strip with singular potentials, J.~Funct.~Anal.~{\bf 94} (1990), 135--155

\bibitem{KoSi1988} S.\ Kotani and B.\ Simon, Stochastic Schr\"odinger operators and Jacobi matrices on the strip, Comm.~Math.~Phys.~{\bf 119} (1988), 403--429

\bibitem{LSM1961} E.\ Lieb, T.\ Schultz and D.\ Mattis, Two soluble models of an antiferromagnetic chain, Annals of Physics~{\bf 16} (1961), 407--466

\bibitem{OH} V.\ Oganesyan and D.\ A.\ Huse, Localization of interacting fermions at high temperature, Phys.~Rev.~B~{\bf 75} (2007), 155111 

\bibitem{PalHuse} A.\ Pal and D.\ A.\ Huse, The many-body localization phase transition, Phys.~Rev.~B {\bf 82} (2010), 174411 

\bibitem{PF1992} L.\ Pastur and A.\ Figotin, Spectra of Random and Almost-Periodic Operators, Springer-Verlag, Berlin, 1992

\bibitem{Pf1970} P.\ Pfeuty, The one-dimensional Ising model with a transverse field, Annals of Physics {\bf 57} (1970), 79--90

\bibitem{SchuBa1} H.\ Schulz-Baldes, Rotation numbers for Jacobi matrices with matrix entries, Math.~Phys.~Electron.~J.~{\bf 13} (2007), Paper 5, 40 pp.

\bibitem{SchuBa2} H.\ Schulz-Baldes, Geometry of Weyl theory for Jacobi matrices with matrix entries, J.~Anal.~Math.~{\bf 110} (2010), 129--165

\bibitem{Stolz2011} G.\ Stolz, An introduction to the mathematics of Anderson localization. Entropy and the quantum II, 71--108, Contemp.~Math.~{\bf 552}, Amer.~Math.~Soc., Providence, RI, 2011

\bibitem{ZPP} M.\ Znidaric, T.\ Prosen and P.\ Prelovsek, Many-body localization in the Heisenberg XXZ magnet in a random field, Phys.~Rev.~B {\bf 77} (2008), 064426 



\end{thebibliography}
\end{document}